\documentclass[conference]{IEEEtran}
\IEEEoverridecommandlockouts
\usepackage{cite}
\usepackage{amsmath,amssymb,amsfonts}
\usepackage{amsthm}
\usepackage{graphicx}
\usepackage{subcaption}
\graphicspath{{figs/}}
\usepackage{tikz}
\usetikzlibrary{positioning,arrows.meta,calc}
\usepackage{textcomp}
\usepackage{xcolor}

\usepackage{booktabs}
\usepackage{multirow}
\usepackage{url}
\usepackage{enumitem}
\usepackage{hyperref}
\hypersetup{colorlinks=true, linkcolor=blue, citecolor=blue, urlcolor=cyan}
\def\BibTeX{{\rm B\kern-.05em{\sc i\kern-.025em b}\kern-.08em
    T\kern-.1667em\lower.7ex\hbox{E}\kern-.125emX}}

\newtheoremstyle{myCustomStyle}{2pt}{2pt}{\normalfont}{0pt}{}{.}{0.5em}{}
\theoremstyle{myCustomStyle}
\newtheorem{myDef}{\textsc{Definition}}
\newtheorem{myTheo}{\textsc{Theorem}}

\newtheorem{myLem}{\textsc{Lemma}}
\newtheorem{myProp}{\textsc{Proposition}}

\begin{document}

\title{VeriTS: Verifiable Model-Enhanced Time-Series Queries on Blockchain Systems}

\author{% five authors, with college line, default IEEE flow
\IEEEauthorblockN{1\textsuperscript{st} Zhongming Yao\textsuperscript{*}\thanks{\textsuperscript{*}Corresponding author.}}
\IEEEauthorblockA{\textit{College of Computer Science and Technology} \\
\textit{Zhejiang University}\\
Hangzhou, China \\
zyao@ieee.org}
\and[\leftskip=-10pt\rightskip=-10pt\hfill]
\IEEEauthorblockN{2\textsuperscript{nd} Jun Pang}
\IEEEauthorblockA{\textit{College of Computer Science and Technology} \\
\textit{Wuhan University of Science and Technology}\\
Wuhan, China \\
pangjun@wust.edu.cn}
\and
\IEEEauthorblockN{3\textsuperscript{rd} Chenxu Wang}
\IEEEauthorblockA{\textit{School of Software Engineering} \\
\textit{Xi'an Jiaotong University}\\
Xi'an, China \\
cxwang@mail.xjtu.edu.cn}
\and[\hfill\mbox{}\newline\mbox{}\hfill]
\IEEEauthorblockN{4\textsuperscript{th} Qian Ma}
\IEEEauthorblockA{\textit{Information Science and Technology College} \\
\textit{Dalian Maritime University}\\
Dalian, China \\
maqian@dlmu.edu.cn}
\and
\IEEEauthorblockN{5\textsuperscript{th} Peiyuan Guan}
\IEEEauthorblockA{\textit{School of Computer Science \& School of Software} \\
\textit{Nanjing University of Information Science \& Technology}\\
Nanjing, China \\
peiyuang@nuist.edu.cn}
\and
\IEEEauthorblockN{6\textsuperscript{th} Shiliang Zhang}
\IEEEauthorblockA{\textit{Department of Informatics} \\
\textit{University of Oslo}\\
Oslo, Norway \\
shilianz@ifi.uio.no}
% ---- one-line version (no college line, four authors, pre-Guan), kept for easy restore ----
% \IEEEauthorblockN{1\textsuperscript{st} Zhongming Yao\textsuperscript{*}\thanks{\textsuperscript{*}Corresponding author.}}
% \IEEEauthorblockA{\textit{Zhejiang University}\\
% Zhejiang, China \\
% zyao@ieee.org}
% \and
% \IEEEauthorblockN{2\textsuperscript{nd} Qian Ma}
% \IEEEauthorblockA{\textit{Dalian Maritime University}\\
% Dalian, China \\
% maqian@dlmu.edu.cn}
% \and
% \IEEEauthorblockN{3\textsuperscript{rd} Chenxu Wang}
% \IEEEauthorblockA{\textit{Xi'an Jiaotong University}\\
% Xi'an, China \\
% cxwang@mail.xjtu.edu.cn}
% \and
% \IEEEauthorblockN{4\textsuperscript{th} Shiliang Zhang}
% \IEEEauthorblockA{\textit{University of Oslo}\\
% Oslo, Norway \\
% shilianz@ifi.uio.no}
}

\maketitle

\begin{abstract}
Every blockchain transaction carries a timestamp, and the chain imposes a total order. On-chain data therefore forms per-source time-series streams. However, existing systems support only basic lookups on blocks and transactions, and cannot answer time-series queries such as time-range retrieval and windowed aggregation. Offloading queries off-chain restores expressiveness, but the off-chain query layer is untrusted, so results must be verifiable. To this end, we propose \texttt{VeriTS}, the first \underline{\textbf{veri}}fiable \underline{\textbf{t}}ime-\underline{\textbf{s}}eries query framework for blockchain systems. It supports efficient range and aggregation queries. \texttt{VeriTS} maintains an off-chain query layer. In this layer, each stream is kept under one tree whose nodes carry authenticated aggregates, so the query index is itself the authenticated data structure. A light client thus verifies a windowed aggregate from a logarithmic number of authenticated nodes rather than from every record. \texttt{VeriTS} further answers error-tolerant queries from compact model representations of a stream, and extends the completeness and soundness guarantees to such approximate answers. As \texttt{VeriTS} never trusts the model behind a representation, a faulty or adversarial model can only widen an answer's certified interval, never falsify it. Experiments offer evidence that on windowed aggregation, \texttt{VeriTS} improves verification efficiency by more than two orders of magnitude over per-record proofs. On range retrieval, proofs shrink by up to $14.5\times$.
\end{abstract}


% ======================================================================
\section{Introduction}\label{sec:intro}

Blockchain is a decentralized, distributed ledger technology. Through hashed data structures, cryptographic primitives, and consensus mechanisms, it keeps data secure, transparent, and immutable. It is now used widely in finance~\cite{treleaven2017blockchain}, healthcare~\cite{attaran2022blockchain}, and energy~\cite{liu2026clustered}. Beyond its role as a ledger, blockchain data is inherently temporal. Every transaction carries a timestamp, and the chain imposes a total order. On-chain data thus forms time-series streams appended block after block, such as an account's transfer amounts, a trading pair's successive prices, or a contract's gas consumption. These ordered streams constitute the temporal view of on-chain data. Concretely, the record $o_{r}=(t_r, v_r)$ extracted from a transaction in block $B_r$ carries a timestamp $t_r$ and a value $v_r$, such as a transferred amount. Grouped by source and ordered by time, such records form time-ordered streams.

Applications continuously query such streams. A regulator audits the sum of an account's transfers over a reporting window. A trader requests the min/max trade price in an interval. A protocol monitor asks for the average gas price over a sliding range. These reduce to two primitives, time-range retrieval and windowed aggregation over SUM, COUNT, MIN, MAX, and AVG. However, blockchains support only block- and transaction-level lookups~\cite{li2020survey}. Under sequential storage, even a single windowed aggregate forces a scan of every block in the window. As the chain grows to millions of blocks, every such query degrades to a full range scan.

We aim to enable expressive time-series queries on blockchains. This raises three challenges.

\noindent\textbf{\textit{Challenge 1: How to organize blockchain time-series data for efficient querying.}} Existing studies offload blockchain data to high-performance external databases~\cite{li2017etherql,zhu2019sebdb,peng2020falcondb,wu2021vql,yao2025vgq}, but they target relational or graph workloads. A relational scheme flattens transactions into rows and tables, so a windowed aggregate either scans the window or maintains auxiliary indexes. A graph scheme serves topology traversals and has no native index over an ordered time interval. An on-chain stream moreover keeps appending block after block, and this write-heavy workload demands an index that keeps up with appends cheaply. Neither line of work indexes windowed aggregation over an ordered interval, the dominant time-series access pattern, without auxiliary scans.

\noindent\textbf{\textit{Challenge 2: How to achieve efficient light-client verification of time-series query results.}} Whatever the off-chain scheme, data managed by the query layer is not trustworthy, for two reasons. First, unlike a blockchain, an off-chain store lacks tamper-proof structures, so its data can be silently altered. Second, the query-layer operator is outside the blockchain consensus and may behave dishonestly, compromising completeness or soundness. Hence results must be verified. For an aggregate, verification means two things. Completeness requires that no in-range point is missing, so that the aggregate covers exactly the queried window. Soundness requires that every returned point lies on-chain and in-range. Re-downloading the window and recomputing~\cite{wu2021vql} defeats the purpose for a light client. Without an authenticated data structure (ADS), the client cannot detect a dropped or fabricated record.

\noindent\textbf{\textit{Challenge 3: How to make approximate answers from model representations verifiable.}} Time series also offer an opportunity absent from the graph and relational settings. Monitoring queries, such as the sliding-window gas average above, tolerate a bounded error yet are issued continuously and demand low latency. Time-series mining~\cite{li2022evolutionary,hu2023spatio,hu2024estimator,yao2024tsec} increasingly runs on compact model-based representations, from piecewise-linear approximation~\cite{eichinger2015time} to today's learned representations. Answering such queries from model parameters instead of records is orders of magnitude cheaper. Two obstacles stand in the way. The first is verification semantics. An answer computed from a representation carries no evidence that it reflects the data it claims to summarize. To our knowledge, no prior verifiable-query system defines what completeness and soundness mean when a value is reconstructed within an error bound rather than retrieved. The second obstacle, more fundamental once learned representations enter the picture, is trust. A model is an opaque, expensive, and often non-deterministic artifact~\cite{xu2026tcrl}. Admitting one into a verifiable pipeline normally means enlarging the trusted computing base or asking a light client to re-run inference. Neither is tenable.

To address these three challenges we propose \texttt{VeriTS}. For the first, \texttt{VeriTS} establishes an off-chain query layer that models on-chain data as ordered streams, where each record carries a timestamp, a value, and tags. This is the native model of production time-series databases~\cite{wang2023iotdb}, and it enables logarithmic range and aggregation queries. For the second, \texttt{VeriTS} introduces a layered verification mechanism. The query layer answers queries. Miners on the blockchain layer authenticate the ADS and anchor its digest on-chain, so a light client verifies with a succinct verification object. For the third, \texttt{VeriTS} adds a verifiable approximate path that answers error-tolerant queries from authenticated model segments. Their digest is likewise anchored on-chain. The approximate path is an addition, not a replacement. The regulator's audit still demands a bit-exact sum. Heavy-tailed streams of transfer events compress too poorly for model segments to pay off. \texttt{VeriTS} therefore maintains both paths and lets the client choose per query. Concretely, the main contributions are summarized as follows.

\begin{itemize}[leftmargin=*]
\item We present \texttt{VeriTS}, to our knowledge the first framework for verifiable time-series queries on blockchains. It gives blockchain data the query interface of a time-series database, and every answer can be verified against the blockchain.

\item \texttt{VeriTS} represents blockchain data as per-source ordered streams over an authenticated aggregate interval tree. The query index is itself the ADS, so the tree both answers and certifies windowed aggregates.

\item \texttt{VeriTS} verifies completeness through a minimum covering set with boundary proofs and soundness through subtree-embedded aggregate folding. Both extend to range and tag-constrained queries and to sliding windows.

\item \texttt{VeriTS} further adds a verifiable approximate path over model segments, redefining completeness and soundness under bounded error. An untrusted encoder affects proof size and answer width but never correctness.

\item An experimental study offers evidence that \texttt{VeriTS} is capable of outperforming baselines in terms of verification efficiency by more than two orders of magnitude and of reducing range-query proofs by up to $14.5\times$.
\end{itemize}

% [REMOVED 2026-09-06 for space: roadmap paragraph; restore by deleting the leading % on the next line]
% Section~\ref{sec:related} reviews related work. Section~\ref{sec:pre} gives preliminaries. Section~\ref{sec:framework} presents the framework. Section~\ref{sec:verify} details verifiable query processing. Section~\ref{sec:fastpath} presents the verifiable approximate path over model segments. Section~\ref{sec:security} analyzes security. Section~\ref{sec:exp} reports the experimental results. Section~\ref{sec:conclusion} concludes.

% ======================================================================
\section{Related Work}\label{sec:related}

\noindent\textbf{Querying with authenticated data structures.}
Verifiable queries on blockchains have been built on accumulator- and Merkle-based ADSs. vChain~\cite{xu2019vchain} and vChain+~\cite{wang2022vchain+} answer Boolean range queries over object tuples with accumulator indexes and boundary proofs. GEM$^2$-tree~\cite{zhang2019gem} reduces on-chain gas for authenticated range queries. Several schemes~\cite{ruan2021lineagechain,singh2023efficient,zhang2021authenticated,wen2023fine} extend ADSs to provenance, range, keyword, and sharded access. IntegriDB~\cite{zhang2015integridb} builds an accumulator-based ADS for a subset of SQL, returning the actual aggregate rather than mere set membership. For aggregation, the authenticated aggregation B-tree~\cite{li2010authenticated} embeds aggregate values and Merkle hashes in index nodes to certify SUM/COUNT/MIN/MAX/AVG and even holistic aggregates. However, none of these schemes targets the ordered, high-write time-series stream, and all authenticate exact values.

\noindent\textbf{On-chain/off-chain models.}
Hybrid systems store bulk data off-chain and anchor digests on-chain. SEBDB~\cite{zhu2019sebdb}, FalconDB~\cite{peng2020falcondb}, and VeriDKG~\cite{zhou2023veridkg} support efficient off-chain querying with succinct on-chain metadata. VeriDKG splits each proof into a Merkle proof for completeness and an aggregation proof for soundness. Nevertheless, these systems target relational/knowledge-graph data, not time series.

\noindent\textbf{Querying through query layers.}
EtherQL~\cite{li2017etherql} and VQL~\cite{wu2021vql} offload blockchain queries to an external database, with VQL anchoring a per-database fingerprint on-chain and splitting heavy miner-side verification from light user-side checking. Blockchain databases built in this mold further support top-$k$ transaction-path queries~\cite{hao2022efficient}, learned-index-based semantic keyword search~\cite{yao2023learned}, trusted data provenance~\cite{yao2023efficient}, and secure sharing across interoperable chains~\cite{hao2023efficient}. Most recently, VGQ~\cite{yao2025vgq} authenticates queries over the graph view of on-chain data, and VTRQ~\cite{yao2026vtrq} over its trajectory view. However, per-query integrity in these layers is absent, coarse-grained, or bound to the graph or trajectory view. Verifiable windowed aggregation over the temporal view is left open.

In summary, existing schemes authenticate exact values over relational, object, or graph data. \texttt{VeriTS} instead certifies windowed aggregates over ordered on-chain streams from one tree whose query index is itself the ADS, and extends the guarantee to error-bounded answers.

% ======================================================================

\section{Preliminaries}\label{sec:pre}

\subsection{Cryptographic Hash Function}
\begin{myDef}\label{def:hash}
A \textbf{cryptographic hash function} $\textit{hash}(\cdot)$ maps input $d$ to a fixed-size digest $\textit{hash}(d)$ and satisfies (i) preimage resistance, (ii) second-preimage resistance, and (iii) collision resistance.
\end{myDef}

\subsection{Blockchain}
\begin{myDef}\label{def:bc}
A \textbf{blockchain} $\textit{BC}=\{B_r\mid r\geq 1\}$ is a sequence of blocks. A \textbf{block} $B_r=(S_{\textit{header}},S_{\textit{body}})$ has header $S_{\textit{header}}=(r,\textit{hash}(B_{r-1}),t_r,\textit{root}_r)$ and body $S_{\textit{body}}$ containing a list of transactions, where $t_r$ is the timestamp when $B_r$ was created and $\textit{root}_r$ is the root of an on-chain ADS. Once appended, $B_r$ is immutable.
\end{myDef}

\begin{myDef}\label{def:tx}
A \textbf{transaction} in block $B_r$ is $\textit{tx}_{r,i}=(i,\textit{sid},t,v,W)$, where $i$ is the index of the transaction within $B_r$, $\textit{sid}$ identifies its source (e.g., the sender account, a contract, or a trading pair), $t$ is its timestamp, i.e., the block timestamp $t_r$ or the transaction's own creation time where the chain records one, $v$ is a numeric value carried by the transaction (e.g., the transferred amount), and $W$ is an optional set of further attributes (e.g., the transaction type).
\end{myDef}

% ======================================================================

\section{The VeriTS Framework}\label{sec:framework}

\subsection{Overview}\label{sec:overview}
\begin{figure}[t]
\centering
\includegraphics[width=\columnwidth]{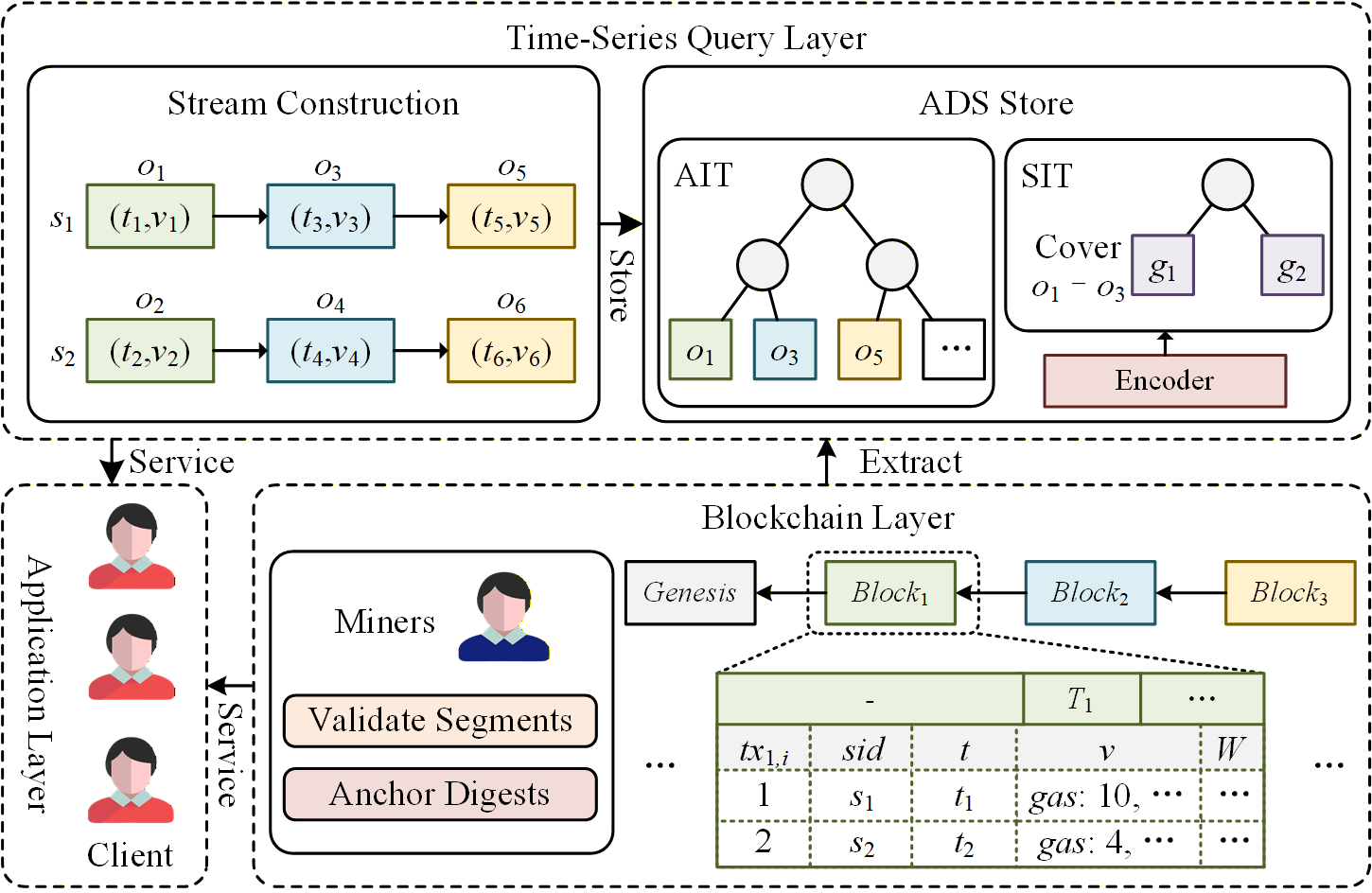}
\caption{The framework of \texttt{VeriTS}.}
\label{fig:framework}
\vspace{-6mm}
\end{figure}
Fig.~\ref{fig:framework} illustrates the framework of \texttt{VeriTS}. It involves three parties as follows.

\noindent\textbf{Blockchain Layer.} The blockchain layer stores transaction data. Miners in the blockchain provide the verification service.

\noindent\textbf{Query Layer.} The query layer is an off-chain service that ingests records from the blockchain and answers queries with a verification object (VO). It keeps each stream in two representations. Raw records under the authenticated aggregate interval tree (AIT) of Section~\ref{sec:ads} serve exact queries, the \emph{exact path}. Model segments under the segment-interval tree (SIT) of Section~\ref{sec:segauth} serve approximate queries with error-bounded answers, the \emph{approximate path}.

\noindent\textbf{Application Layer.} The application layer is a light client that issues queries and can verify the query results. It does not store raw data.

\texttt{VeriTS} operates in five steps, each carried by a component of Fig.~\ref{fig:framework}. First, the query layer extracts records from appended blocks (Extract) and groups them into time-ordered streams (Stream Construction). Second, off-line, the Encoder, the only learned component in \texttt{VeriTS}, partitions each stream into model segments and assigns their error budgets. Third, records and segments enter the AIT and the SIT of the ADS Store (Store). Miners validate the segments against the raw records (Validate Segments), fold each stream's two roots into one forest root, and write it into the block header under consensus (Anchor Digests). Fourth, the client sends $q$ to the query layer, which returns the answer with a VO (Service, query layer). Fifth, the client checks the VO against the root anchored on the chain (Service, blockchain layer) and accepts or rejects.

\noindent\textbf{Threat Model.} The blockchain layer is trusted as a whole, since Byzantine fault-tolerant consensus makes an anchored digest correct and immutable. The application layer is trusted only for its own verification. The query layer, including its encoder, is untrusted. A polynomial-time adversary controlling it may alter, drop, or fabricate records, submit segments with false parameters or bounds, and forge or truncate a VO. It cannot forge a block header or find hash collisions.

\subsection{Data Model and Queries}\label{sec:datamodel}\label{sec:queries}
\texttt{VeriTS} organizes transaction data into time-series form.

\begin{myDef}\label{def:record}
A \textbf{time-series record} is $o=(\textit{sid},t,v,W)$, where $\textit{sid}$ identifies the source stream, $t$ is the timestamp, $v$ is the value, and $W$ is an optional attribute (tag) set.
\end{myDef}

Each record is derived from the transactions of one block $B_r$ and carries an on-chain position within that block. A record arises in one of two ways: (i) it is a transaction of Definition~\ref{def:tx} itself, with position $(r,i)$; or (ii) it is a per-block metric computed from all transactions of $B_r$, such as its mean gas price, with position $(r,0)$.

\begin{myDef}\label{def:stream}
A \textbf{time-series stream} $S=\langle o_1,\dots,o_n\rangle$ collects the $n$ records of one $\textit{sid}$ in increasing order of $t$, ties broken by on-chain position, so the order within a stream is strict. The on-chain time-series dataset is $\mathcal{S}=\{S^{(1)},\dots,S^{(m)}\}$, one stream per source.
\end{myDef}

\texttt{VeriTS} supports two query types over a stream, as follows.
\begin{myDef}\label{def:rangeq}
A \textbf{time-range query} $q_{\textit{range}}=(\textit{sid},[t_s,t_e])$ returns all records of stream $\textit{sid}$ with $t\in[t_s,t_e]$.
\end{myDef}
\begin{myDef}\label{def:aggq}
An \textbf{aggregation query} $q_{\textit{agg}}=(\textit{sid},[t_s,t_e],\textsf{f})$ returns $\textsf{f}(\{v : t\in[t_s,t_e]\})$ for $\textsf{f}\in\{\textsf{SUM},\textsf{COUNT},\textsf{MIN},$ $\textsf{MAX},\textsf{AVG}\}$.
\end{myDef}

% [REMOVED 2026-09-06 for space: Example 1 (ex:aggq); restore by deleting the leading % on the following lines]
% \begin{myEx}\label{ex:aggq}
% Suppose stream $\textit{sid}$ holds five records $o_r=(t_r,v_r)$ with values $10,12,9,15,11$ at times $t_1,\dots,t_5$. Then $q_{\textit{agg}}=(\textit{sid},[t_1,t_3],\textsf{SUM})$ returns $10+12+9=31$. With $\textsf{AVG}$ it returns $31/3$.
% \end{myEx}

% ======================================================================

\section{Verifiable Query Processing}\label{sec:verify}

\subsection{Authenticated Aggregate Interval Tree}\label{sec:ads}
\texttt{VeriTS}'s ADS is an \emph{authenticated aggregate interval tree} (AIT) over the leaf order of a stream. It adapts the authenticated aggregation B-tree~\cite{li2010authenticated} to append-only, time-ordered streams with a chain-anchored digest. A leaf holds one record, hashed as $\textit{hash}(t\,\|\,v\,\|\,W\,\|\,(r,i))$ with aggregate $(v,1,v,v)$. Each internal node $u$ stores an aggregate $\alpha_u=(\textsf{sum},\textsf{count},\textsf{min},\textsf{max})$ of its subtree, folded componentwise from its $f$ children $u_1,\dots,u_f$, and a hash
\begin{equation}
\eta_u=\textit{hash}(\eta_{u_1}\,\|\,\alpha_{u_1}\,\|\,\sigma_{u_1}\,\|\,\cdots\,\|\,\eta_{u_f}\,\|\,\alpha_{u_f}\,\|\,\sigma_{u_f}),
\label{eq:node}
\end{equation}
where $\sigma_u=t^{\min}_u\,\|\,t^{\max}_u$ is the time span a node $u$ covers, so an internal node certifies the aggregate and the span of its whole subtree without traversing leaves. The root hash $\eta_{\textit{root}}$ is the per-stream digest anchored in the block header of Definition~\ref{def:bc}. With $m$ streams, miners fold the per-stream digests into one Merkle forest whose root is the $\textit{root}_r$ of Definition~\ref{def:bc}. Incremental maintenance on append updates the $O(\log n)$ nodes on one root-to-leaf path.

The AIT is also the query layer's primary index. A windowed aggregate folds the aggregates of the $O(\log n)$ nodes covering the window instead of scanning it, so acceleration and verifiability come from one structure rather than a separate index and ADS.

\noindent\textbf{Completeness Verification.} A query answer is \emph{complete} iff no on-chain record of $\textit{sid}$ with $t\in[t_s,t_e]$ is omitted. Completeness is certified by a minimum covering set (MCS), a set of disjoint subtrees whose leaf ranges union to exactly $[t_s,t_e]$. The VO adds their sibling hashes (SIB) and the two boundary records just outside the window, the last before $t_s$ and the first after $t_e$. The client recomputes $\eta_{\textit{root}}$ from the MCS, the two boundary leaves, and SIB via Eq.~\eqref{eq:node}, matches it against the anchored digest, and checks that every MCS node's authenticated span lies within $[t_s,t_e]$. Each boundary leaf is authenticated adjacent to an MCS edge in leaf order, so no in-window record can hide at either edge. The anchored digest comes from a block whose timestamp exceeds $t_e$.

% [PENDING-RESTORE 2026-08-28] 完整性运行例子，篇幅允许时恢复（去掉下面各行行首的 %）
% \begin{myEx}\label{ex:complete}
% Continuing Example~\ref{ex:aggq}, for $[t_1,t_3]$ the MCS is the subtree covering $\{o_1,o_2,o_3\}$. The boundary record $o_4$ (first after $t_e=t_3$) proves nothing between $o_3$ and $o_4$ was hidden. Recomputing the root and matching the anchored digest confirms completeness. Dropping $o_2$ changes the subtree hash and fails the root check.
% \end{myEx}

\noindent\textbf{Soundness Verification.} A query answer is \emph{sound} iff every returned record is authenticated and in-range, and any returned aggregate equals the verified computation over exactly the in-range records. Because each MCS node carries its subtree aggregate $\alpha_u$ and hash $\eta_u$, the client folds the MCS aggregates componentwise, $\textsf{ANS}(q)=\sum_{u\in\textit{MCS}(q)}\alpha_u$, and accepts only if the folded root matches the anchored digest. For $\textsf{MIN}/\textsf{MAX}$, summation in the fold is replaced by $\min/\max$. Since Eq.~\eqref{eq:node} binds every $\alpha_u$ into its parent's hash, an understated subtree extremum is caught by the root check exactly as a tampered sum. $\textsf{AVG}$ is verified as $\textsf{SUM}/\textsf{COUNT}$ over one shared MCS. The client thus runs three checks in order, the covering check, the root check against the anchored digest, and the $\alpha$-fold.

\subsection{Range and Tag-Constrained Queries}\label{sec:rangetag}
For a time-range query, the VO opens the MCS subtrees down to their leaves, so the query layer returns all $k$ in-window records with SIB and the two boundary records. The client rebuilds the leaf hashes, recomputes $\eta_{\textit{root}}$, and matches the anchored digest. Completeness follows from the MCS-and-boundary argument of Section~\ref{sec:ads} and soundness from each record's authenticated inclusion, at cost $O(k+\log n)$, linear only in the result the client must receive anyway.

A tag-constrained aggregation applies $\textsf{f}$ to the in-window records whose tag set $W$ satisfies a predicate. It cannot use the $\alpha$-fold, since subtree aggregates summarize all records in a subtree rather than the filtered subset. \texttt{VeriTS} therefore answers it through the range VO. The client filters the authenticated in-window records by the predicate locally and aggregates the subset, inheriting soundness from leaf authentication and completeness from the MCS, again in $O(k+\log n)$.

\vspace{-0.5mm}
\subsection{Sliding-Window Aggregation}\label{sec:sliding}
\vspace{-0.5mm}

Monitoring applications, such as the gas-price monitor of Section~\ref{sec:intro}, issue a sequence of aggregation queries over windows $W_i=[t_s+i\delta,\ t_e+i\delta]$ that overlap heavily. Verifying each window independently repeats work. Adjacent MCSs share all nodes covering $W_i\cap W_{i+1}$. \texttt{VeriTS} exploits this by letting the client cache the verified $(\eta_u,\alpha_u)$ pairs of the current MCS. On each slide the query layer ships only the \emph{delta VO}. It contains the nodes covering the entering sub-interval $W_{i+1}\setminus W_i$ and the new boundary record. Each slide is pinned to an explicit block height and verified against that block's anchored digest. The client drops expired nodes and re-folds. Each slide thus verifies $O(\log n)$ new nodes in the worst case, and only the entering records' cover in the common case, amortizing VO bandwidth across the query sequence.

% [PENDING-RESTORE 2026-08-28] Cost Analysis 小节（Lemma+证明+收尾段），篇幅允许时恢复（去掉下面各行行首的 %）
% \subsection{Cost Analysis}\label{sec:cost}
% \begin{myLem}\label{lem:cost}
% For an aggregation query over a window of $k$ records in an AIT of $n$ leaves with fanout $f$, the MCS contains at most $2\lceil\log_f n\rceil$ nodes. The VO (MCS, SIB, and two boundary records) has size $O(f\log_f n)$. VO generation takes $O(\log n)$ after two root-to-leaf boundary descents. Client verification takes $O(f\log_f n)$ hash evaluations and fold operations, independent of $k$.
% \end{myLem}
% \begin{proof}
% The window's leaf interval admits a canonical cover with at most two maximal subtrees per level of the tree, giving $|\textit{MCS}|\leq 2\lceil\log_f n\rceil$. Each MCS or path node contributes at most $f{-}1$ sibling entries to SIB, giving the $O(f\log_f n)$ VO size. The server locates the two window boundaries by two descents and emits the cover in the same traversal. The client performs one hash per reconstructed node on the two root paths and one fold operation per MCS node.
% \end{proof}
% For distributive aggregates the structure is hash-only (no pairing operations), with KB-scale proofs expected at $10^5$--$10^6$ records~\cite{li2010authenticated,zhang2015integridb}. Range and tag-constrained queries add the unavoidable $O(k)$ result transfer of Section~\ref{sec:rangetag}.

% ======================================================================
\vspace{-0.5mm}
\section{Verifiable Approximate Path}\label{sec:fastpath}
\vspace{-0.5mm}
Section~\ref{sec:verify} returns bit-exact answers, but many applications tolerate a bounded error in exchange for latency~\cite{eichinger2015time,li2020compression,li2021trace}, e.g., tracking a dashboard's average gas price. \texttt{VeriTS} therefore maintains \emph{model segments} for each stream, answering queries from segment metadata within a certified error bound. Learning enters only through the encoder, which cuts a stream into segments and allocates their error budgets (Section~\ref{sec:contract}). Certification (Section~\ref{sec:segauth}) and verification (Section~\ref{sec:approxsem}) are closed-form and never evaluate the encoder.

\subsection{Model Segments}\label{sec:contract}
\begin{myDef}\label{def:encoder}
An \textbf{encoder} $\mathcal{E}$ maps a stream $S$ to an ordered set of \textbf{model segments} $\{g_1,\dots,g_s\}$ partitioning its time domain, where $g_j=(\,[\tau_j,\tau_{j+1}),\,\theta_j,\,\phi_j,\,$ $\varepsilon^{v}_j,\,\varepsilon^{t}_j,\,c_j\,)$ represents the $c_j$ records with $t\in[\tau_j,\tau_{j+1})$ in position order $p=0,\dots,c_j-1$. \texttt{VeriTS} imposes only three \textbf{representation constraints} on $\mathcal{E}$. (i)~\emph{Value model.} The covered values are reconstructed by $\hat v_j(p)$ with parameters $\theta_j$ from a fixed closed-form family $\mathcal{F}$, such that $\textsf{f}$ over the reconstructed values of any contiguous position range is computable from $\theta_j$ in $O(1)$. (ii)~\emph{Arrival model.} The covered timestamps are reconstructed by a monotone closed-form $\hat t_j(p)$ with parameters $\phi_j$, so that a queried time range maps to a position range in $O(1)$. (iii)~\emph{Certified bounds.} For every covered record,
\begin{equation}
|\,\hat v_j(p)-v_p\,|\leq\varepsilon^{v}_j,\qquad |\,\hat t_j(p)-t_p\,|\leq\varepsilon^{t}_j.
\label{eq:bounds}
\end{equation}
\end{myDef}

A window boundary rarely coincides with a segment boundary, so answering $[t_s,t_e]$ requires knowing which records of the two straddled segments fall inside. The value model alone does not carry that information. Condition~(ii) of Definition~\ref{def:encoder} models the arrival process as well and certifies it by the same mechanism, which keeps that resolution closed-form (Section~\ref{sec:approxsem}). Nothing else about $\mathcal{E}$ is assumed. Where to cut the stream, with which coefficients, and how to split the error budget are left to $\mathcal{E}$. \texttt{VeriTS} thus separates a learning problem, solved off-line by a model of any size, from a verification problem, solved on-line by closed-form arithmetic.

We instantiate both models in the linear family, $\hat v_j(p)=a_jp+b_j$ and $\hat t_j(p)=\mu_jp+\beta_j$ with $\mu_j>0$ the mean inter-arrival time. In this family \textsf{SUM}, \textsf{COUNT}, \textsf{MIN}, \textsf{MAX}, and \textsf{AVG} over contiguous position ranges are closed-form.

The stream is stored as an ordered set of segments $\{g_1,\dots,g_s\}$ with $s\ll n$~\cite{yao2024camel}. Because both models are indexed by position rather than by time, an aggregate over a contiguous position range is a function of $\theta_j$ and the range endpoints alone. No per-record data is needed. A window is therefore answered in time proportional to the number of overlapping segments rather than records.

\vspace{-1mm}
\subsection{Segment Authentication}\label{sec:segauth}
\vspace{-1mm}

Segments are authenticated by a segment-interval tree (SIT), an AIT built as in Eq.~\eqref{eq:node} with leaves $\textit{hash}(\,[\tau_j,\tau_{j+1})\,\|\,\theta_j\,\|\,$ $\phi_j\,\|\,\varepsilon^{v}_j\,\|\,\varepsilon^{t}_j\,\|\,c_j\,)$, whose root is anchored on-chain alongside the AIT root. It therefore inherits the AIT's shape, hashing rule, and covering-set and boundary-proof checks of Section~\ref{sec:ads}, and differs in two respects. (i)~\emph{Leaf.} An AIT leaf is one record with its exact value and a one-timestamp span $\sigma$. A SIT leaf is one segment, whose span $\sigma$ is $[\tau_j,\tau_{j+1})$ and whose aggregate $\alpha$ is the closed-form aggregate of its reconstructions over all $c_j$ positions plus a fifth component, the value budget $\varepsilon^{v}_j c_j$. (ii)~\emph{Fold.} An AIT cover folds to the exact answer. A SIT cover folds to the interior's reconstructed aggregate and its total value slack, the $\Delta_v$ of Lemma~\ref{lem:eps}, in $O(\log s)$, while only the two straddled segments are opened to their model fields. One more scalar is anchored with the SIT root, the stream's budget cap $\max_j\varepsilon^{v}_j$. $\textsf{MIN}/\textsf{MAX}$ verification (Lemma~\ref{lem:eps}) needs the largest budget among the answering segments, which a folded aggregate does not carry, and substituting the cap merely widens the intervals. Before a segment leaf enters the tree, miners, who hold the raw records, replay the covered window and check Eq.~\eqref{eq:bounds} at every position, the record count $c_j$, and abutment, $\tau_j$ equal to the previous leaf's right endpoint. Only validated segments are anchored under consensus. Since each record belongs to exactly one segment, this check costs amortized $O(1)$ per record and lies entirely off the query path. Miners evaluate the two closed forms but never execute $\mathcal{E}$, so certification is independent of the encoder's cost, its determinism, and whether its parameters are public.

\vspace{-2.3mm}
\subsection{Verification Semantics and Error Bound}\label{sec:approxsem}
\vspace{-0.7mm}

An accepted approximate answer carries two guarantees. Both are the exact-path notions of Section~\ref{sec:ads} restated for segments. Model-completeness forbids a missing segment rather than a missing record, and $\varepsilon$-soundness replaces equality with the exact answer by containment in a certified interval.

\noindent\textbf{Model-Completeness.} An approximate-path answer is \emph{model-complete} iff the query interval $[t_s,t_e]$ is fully covered by authenticated segments (a boundary proof over segments, not records). The client verifies this by a segment-level MCS with boundary proof over the SIT.

\noindent\textbf{$\varepsilon$-Soundness.} The answer is \emph{$\varepsilon$-sound} iff it is assembled exclusively from reconstructions of authenticated parameters, each carrying its certified bound. For every position $p$ the answer draws on,
\begin{equation}
v_p\in\big[\hat v_j(p)-\varepsilon^{v}_j,\ \hat v_j(p)+\varepsilon^{v}_j\big]
\text{ and }
t_p\in\big[\hat t_j(p)-\varepsilon^{t}_j,\ \hat t_j(p)+\varepsilon^{t}_j\big].
\label{eq:eps}
\end{equation}
The client computes the interval endpoints of Eq.~\eqref{eq:eps} from the authenticated parameters itself, never receiving a reconstructed value or consulting a raw record.

\noindent\textbf{Boundary Resolution.} A window edge $t_s$ falls inside some segment $g_j$, whose timestamps the client does not hold, so it brackets the in-window positions with the arrival model. Since $t_p\in[\hat t_j(p)-\varepsilon^{t}_j,\hat t_j(p)+\varepsilon^{t}_j]$ and $\hat t_j$ is increasing, position $p$ is certainly in the window when $\hat t_j(p)-\varepsilon^{t}_j\geq t_s$ and certainly out when $\hat t_j(p)+\varepsilon^{t}_j<t_s$. Inverting the linear model, the undecided positions form a run of length at most $k_j=\big\lceil\,2\varepsilon^{t}_j/\mu_j\,\big\rceil$, computed in $O(1)$ from authenticated parameters. The client aggregates over the certain positions and charges the undecided ones to the answer interval. The run depends on the regularity of arrivals, not on segment length. With Ethereum-like spacing, $\mu_j{=}12$\,s and $\varepsilon^{t}_j{=}6$\,s, at most one position is undecided however long the segment. A window must span at least one full segment, and a shorter one falls back to the exact path.

Lemma~\ref{lem:eps} propagates the per-position bounds of Eq.~\eqref{eq:eps} and the boundary charge into an interval for each aggregate. A window is answered over segments $g_{j_1},\dots,g_{j_m}$. Among these, $\partial$ collects the segments not certified to lie wholly inside $[t_s,t_e]$, with undecided total $K=\sum_{j\in\partial}k_j$. The set $\mathcal{C}$ holds the positions certified in the window, with $\tilde c_{j_i}=|\mathcal{C}\cap g_{j_i}|$ per segment, and $\textsf{S}$ is the sum of their reconstructed values. Deviations collect into a value term $\Delta_v=\sum_i\varepsilon^{v}_{j_i}\tilde c_{j_i}$ and a boundary term $\Delta_t=\sum_{j\in\partial}k_jV_j$, where $V_j=\max_p|\hat v_j(p)|+\varepsilon^{v}_j$ is computable from $\theta_j$. For the extrema, $\bar\varepsilon=\max_i\varepsilon^{v}_{j_i}$, while $\hat M$ and $\hat M^{+}$ are the largest reconstructions over $\mathcal{C}$ alone and over $\mathcal{C}$ plus the undecided positions.

\vspace{-1mm}
\begin{myLem}\label{lem:eps}
If every segment satisfies Eq.~\eqref{eq:bounds}, then
\begin{equation}
\begin{aligned}
\textsf{S}^{*}&\in\big[\,\textsf{S}-\Delta_v-\Delta_t,\ \ \textsf{S}+\Delta_v+\Delta_t\,\big],\\
\textsf{COUNT}^{*}&\in\big[\,|\mathcal{C}|,\ |\mathcal{C}|+K\,\big],
\end{aligned}
\label{eq:prop}
\end{equation}
$\textsf{MAX}^{*}\in[\,\hat M-\bar\varepsilon,\ \hat M^{+}+\bar\varepsilon\,]$, symmetrically for $\textsf{MIN}$, and $\textsf{AVG}$ is bounded by dividing the $\textsf{SUM}$ interval by the $\textsf{COUNT}$ interval. If the stream's values are non-negative, as for transferred amounts or gas, the boundary term is one-sided and the lower endpoint tightens to $\textsf{S}-\Delta_v$.
\end{myLem}

\vspace{-2mm}
\begin{proof}
By Eq.~\eqref{eq:bounds} the certified records of $g_{j_i}$ contribute at most $\varepsilon^{v}_{j_i}\tilde c_{j_i}$ deviation, giving $\Delta_v$ in total. The exact answer additionally includes some subset of the undecided positions, at most $k_j$ per segment of $\partial$, each of magnitude at most $V_j$, giving $\Delta_t$, one-sided when values are non-negative. Adding the two bounds gives Eq.~\eqref{eq:prop}, and \textsf{COUNT} questions exactly the undecided positions. For the extrema, the window's true positions include every certified position and are included in the candidate set. Each true value lies within $\bar\varepsilon$ of its reconstruction, which gives both endpoints. $\textsf{MIN}$ is symmetric, and $\textsf{AVG}$ follows from the $\textsf{SUM}$ and $\textsf{COUNT}$ intervals.
\end{proof}
\vspace{-1mm}

The two terms are asymmetric. $\Delta_v$ scales with the window's record count. $\Delta_t$ is paid only at the edges, so long windows are dominated by the value budget and short ones by the boundary budget. Adjacent padded intervals $[\tau_j-\varepsilon^{t}_j,\ \tau_{j+1}+\varepsilon^{t}_j]$ overlap by up to $2\varepsilon^{t}$, so $|\partial|\leq 2$ whenever a segment's time span exceeds $2\varepsilon^{t}$. Lemma~\ref{lem:eps} is stated over $\partial$ so that it stays sound when a loose budget puts several segments on one edge. Our evaluation reports $K$. The client therefore verifies an interval guaranteed to contain the exact answer, rather than a point value of unknown accuracy.

\subsection{Encoder Independence and Payoff}\label{sec:payoff}
The encoder faces two decisions. The first, partitioning under a single global bound, is solved. Greedy extension to the bound already minimizes the segment count~\cite{xie2014optimalplr}, so modeling capacity cannot shrink $s$ further. The second, allocating the two-dimensional budget of condition~(iii) of Definition~\ref{def:encoder}, remains open. Loosening $\varepsilon^{v}_j$ widens the answer interval in proportion to the records the segment contributes. Loosening $\varepsilon^{t}_j$ instead blurs the window edge, at no cost for a segment no window edge ever falls into. A segment interior to many windows therefore wants a tight $\varepsilon^{v}$ and may squander $\varepsilon^{t}$. One that straddles window edges wants the opposite. No local sweep observes the workload. This is the axis on which a learned encoder offers an advantage. Section~\ref{sec:exp-encoder} measures this allocation.
The advantage costs no trust. Let $\mathcal{E}$ be an arbitrary, possibly adversarial encoder.

\vspace{-0.5mm}
\begin{myProp}\label{prop:encoder}
Every approximate-path answer the client accepts returns an interval containing the exact on-chain answer. $\mathcal{E}$ influences the segment count $s$, and through it VO size and latency, but never correctness.
\end{myProp}
\vspace{-2mm}

\begin{proof}
The client's checks read only fields authenticated under the anchored SIT root, and every anchored segment has passed the replay check of Section~\ref{sec:segauth} under consensus. Eq.~\eqref{eq:bounds} therefore holds for every covered record irrespective of provenance, and Lemma~\ref{lem:eps} bounds the deviation of the accepted interval. A segment violating its declared bound is rejected before anchoring. One altered afterwards fails the root check. $\mathcal{E}$ appears in neither check.
\end{proof}
\vspace{-2mm}

The client's trusted computing base is thus unchanged whether the segments were fitted by a two-line greedy sweep or a billion-parameter pre-trained model.

It remains to examine where the approximate path pays. A coarser bound merges more records per segment, trading a smaller segment count $s$ for a wider answer interval.

\noindent\textbf{Aggregation.} The gain is logarithmic and therefore small. Both paths fold a canonical cover, the MCS, so compressing a stream by a factor $n/s$ moves the proof from $O(\log n)$ to $O(\log s)$. It removes $\log_2(n/s)$ levels, three of them at tenfold compression. An approximate path cannot be an order of magnitude cheaper than an exact path that is already logarithmic, and Section~\ref{sec:exp} confirms the effect is modest.

\noindent\textbf{Retrieval.} The gain is the compression ratio itself. A time-range query of Definition~\ref{def:rangeq} obliges the exact path to ship the $k$ records it returns. The approximate path ships instead the $\approx k/(n/s)$ segments that cover them, and the client reconstructs each value within its certified $\varepsilon^{v}_j$. Writing $b_r$ for the bytes of a record and $b_g$ for those of a segment's model fields, the approximate path wins whenever
\begin{equation}
n/s \;>\; b_g / b_r ,
\label{eq:payoff}
\end{equation}
a break-even point set by the stream's compressibility rather than by any parameter of the framework. Below it the approximate path ships more than the exact path, and Section~\ref{sec:exp} exhibits on-chain streams on either side of the threshold.

% ======================================================================
\vspace{-0.5mm}
\section{Security Analysis}\label{sec:security}
\vspace{-0.5mm}

We analyze \texttt{VeriTS} under the threat model of Section~\ref{sec:overview}. Miners maintain both per-stream trees, replay each segment before it enters the SIT (Section~\ref{sec:segauth}), and anchor both roots through the forest root of Section~\ref{sec:ads} under Byzantine fault-tolerant consensus.

\begin{myDef}\label{def:verif}
A time-series query is \textbf{verifiable} if, for every polynomial-time adversary $\mathcal{A}$ controlling the query layer, the probability that the client accepts an answer violating completeness or soundness is negligible. On the exact path these are the record-level notions of Section~\ref{sec:ads}, and on the approximate path the model-completeness and $\varepsilon$-soundness of Section~\ref{sec:approxsem}.
\end{myDef}
\begin{myTheo}\label{theo:verif}
\texttt{VeriTS} queries are verifiable under Definition~\ref{def:verif} if (i) $\textit{hash}(\cdot)$ is collision-resistant and (ii) both anchored digests are correct, i.e., the AIT root matches the on-chain records and the SIT root covers only segments that passed the replay and abutment checks of Section~\ref{sec:segauth}.
\end{myTheo}

\vspace{-2mm}

\begin{proof}
By contradiction. The client's checks of Sections~\ref{sec:ads} and~\ref{sec:approxsem} enforce that the presented cover tiles $[t_s,t_e]$ without gap or overlap, its boundary proof pinning the window edges, and that the root recomputed from the presented nodes equals the anchored digest. Hiding a record or segment therefore opens a gap, which the covering check rejects, or changes a presented node. Tampering with a leaf, an aggregate $\alpha_u$, or a segment field also changes a presented node, since every field enters a leaf hash and Eq.~\eqref{eq:node} binds every $\alpha_u$ into its parent. A wrong total without tampering is excluded because the client folds the values itself. A changed node reproduces the anchored digest only through a hash collision, contradicting~(i), or a forged anchor, contradicting~(ii). On the approximate path, (ii) guarantees Eq.~\eqref{eq:bounds} for every anchored segment, so Lemma~\ref{lem:eps} places the exact answer inside the returned interval with no assumption on the encoder (Proposition~\ref{prop:encoder}). Hence $\mathcal{A}$ succeeds only with negligible probability.
\end{proof}

% ======================================================================
\vspace{-1.5mm}
\section{Experiments}\label{sec:exp}

\subsection{Experimental Settings}\label{sec:exp-settings}
\noindent\textbf{Implementation.} \texttt{VeriTS}'s query layer and application layer are both implemented in Python. All experiments run on an Intel Core i9-12900H with 32\,GB of memory.

\noindent\textbf{Datasets.} Five streams are drawn from 30,000 Ethereum blocks and 200,000 TRON blocks and are summarized in Table~\ref{tab:data}, where $\mu$ is the mean inter-arrival time in seconds, cv the coefficient of variation of the values, and $n/s$ the compression ratio under the default budgets. They span the four kinds of quantity a record can carry.
\begin{itemize}[leftmargin=*]
\item ETH-Gas is a rate, the mean gas price of a block.
\item ETH-Transfer and TRON-Transfer are events, single transferred amounts drawn from a heavy tail.
\item TRON-Balance is a state, the running balance of the account behind TRON-Transfer.
\item TRON-Dust is degenerate, a constant 1 sun (\(10^{-6}\) TRX) from TRON's busiest sender.
\end{itemize}
We keep TRON-Dust as a labeled upper extreme. This is because a stream can compress either through degeneracy or through smoothness, and the two must not be conflated. The TRON-Transfer account is selected by total value moved, not transaction count, which selects dust senders. Ethereum timestamps are block-level, so Ethereum streams are per block. TRON streams are per transaction.

% \noindent\textbf{Granularity.} The two chains are treated differently because their timestamps mean different things, which we checked against the data rather than assumed. Ethereum's \textit{block\_timestamp} is a block timestamp. Every transaction in a block carries it, so 97--99\% of consecutive transaction records are simultaneous, an arrival process no monotone model can describe and no time-range query can separate. Ethereum streams are therefore per block, and a stream's length is the block count. TRON's \textit{timestamp} is the transaction's own creation time. In our sample, 1.51M rows carry 1.51M distinct stamps across 175k blocks, spread across the 3\,s block interval, and no block carries more than one distinct stamp by accident. TRON's arrival process is genuinely continuous, so its streams are kept at transaction granularity and nothing is aggregated away.

% \noindent\textbf{Account selection.} Ranking accounts by transaction count selects spam. The four busiest senders on TRON each move a constant dust amount, giving a coefficient of variation of \(0.000\). Accounts are therefore ranked by total value moved and screened for a non-degenerate value distribution, which yields a genuine economic actor---27,018 transfers totalling \(3.7\times10^{5}\) TRX over more than a hundred distinct amounts.

\begin{table}[t]
\caption{Statistics of the on-chain streams.}
\label{tab:data}
\centering\footnotesize
\begin{tabular}{@{}llrrrr@{}}
\toprule
\textbf{Stream} & \textbf{Kind} & \textbf{Records} & $\boldsymbol{\mu}$ \textbf{(s)} & \textbf{cv} & $\boldsymbol{n/s}$ \\
\midrule
ETH-Gas       & rate       &  29,123 & 13.93 & 0.724 &  2.83 \\
ETH-Transfer  & event      &  29,123 & 13.93 & 6.072 &  2.14 \\
TRON-Transfer & event      &  27,018 & 19.84 & 4.437 &  2.09 \\
TRON-Balance  & state      &  27,018 & 19.84 & 0.573 & 13.17 \\
TRON-Dust     & degenerate & 168,499 &  3.56 & 0.000 & 72.10 \\
\bottomrule
\end{tabular}
\vspace{-4.5mm}
\end{table}

\noindent\textbf{Query workload.} The workload consists of time-range and aggregation queries with $\textsf{f}\in\{\textsf{SUM},\textsf{COUNT},\textsf{MIN},\textsf{MAX},$ $\textsf{AVG}\}$, over windows of 100--2000 records. Sliding sequences use these windows with a slide of 10\%--50\%. The value budget is $\varepsilon^{v}=10\%$ of a stream's median $|v|$. The median resists the heavy tail. The arrival budget $\varepsilon^{t}$ is swept in units of each stream's median inter-arrival time. Every point is the mean of five repetitions over 20--40 random windows.

\noindent\textbf{Baselines.} Query efficiency is compared against the two scan-based placements, and verification against a re-download baseline and three per-technique ablations.

\noindent\underline{\textit{\textbf{QA (Query on Application layer)}}} ships the queried window to the client, which aggregates locally.

\noindent\underline{\textit{\textbf{QB (Query on Blockchain)}}} has miners scan the window on-chain and aggregate there.
Both are scans, differing in where the scan happens and therefore in what crosses the network rather than in complexity.

For verification, \underline{\textit{\textbf{B-V}}} re-downloads the window and recomputes the aggregate~\cite{wu2021vql}, standing in for the VQL-style route.
\texttt{VeriTS} optimizes verification through three techniques. They are (i) minimum-covering-set completeness verification, (ii) subtree-embedded aggregate folding for soundness, and (iii) delta VOs for sliding sequences. Each has a baseline.

\noindent\underline{\textit{\textbf{Completeness baseline (B-MP)}}} proves completeness by a Merkle membership path per in-window record, instead of one minimum covering set with boundary proofs.

\noindent\underline{\textit{\textbf{Soundness baseline (B-Leaf)}}} recomputes the aggregate from authenticated raw leaves, instead of folding subtree aggregates.

\noindent\underline{\textit{\textbf{Sliding baseline (B-Ind)}}} verifies each window of a sliding sequence with an independent VO, instead of shipping only the nodes the client has not already authenticated.

\noindent\textbf{Encoders.} The encoder study of Section~\ref{sec:exp-encoder} compares three budget-allocation policies at a matched segment count, so that VO size and verification time are fixed by construction. The study uses one stream per kind, with TRON-Transfer rather than ETH-Transfer for events because it is the source of TRON-Balance, so the event and state verdicts share identical data. \textit{Uniform-$\varepsilon$} spreads one bound over the whole stream. \textit{Learned} splits the stream into regions and fits the two budgets of Definition~\ref{def:encoder} to a held-out sample of the workload by the water-filling law $\varepsilon_h\propto\omega_h^{-2/3}$, with $\omega_h$ the frequency of queries touching region $h$. \textit{E-Adv} inverts the exponent and spends the budget exactly backwards. The workload places each window on one hotspot region with probability equal to a skew parameter.

\begin{figure}[!t]
\centering
\begin{subfigure}{0.49\linewidth}
  \includegraphics[width=\linewidth]{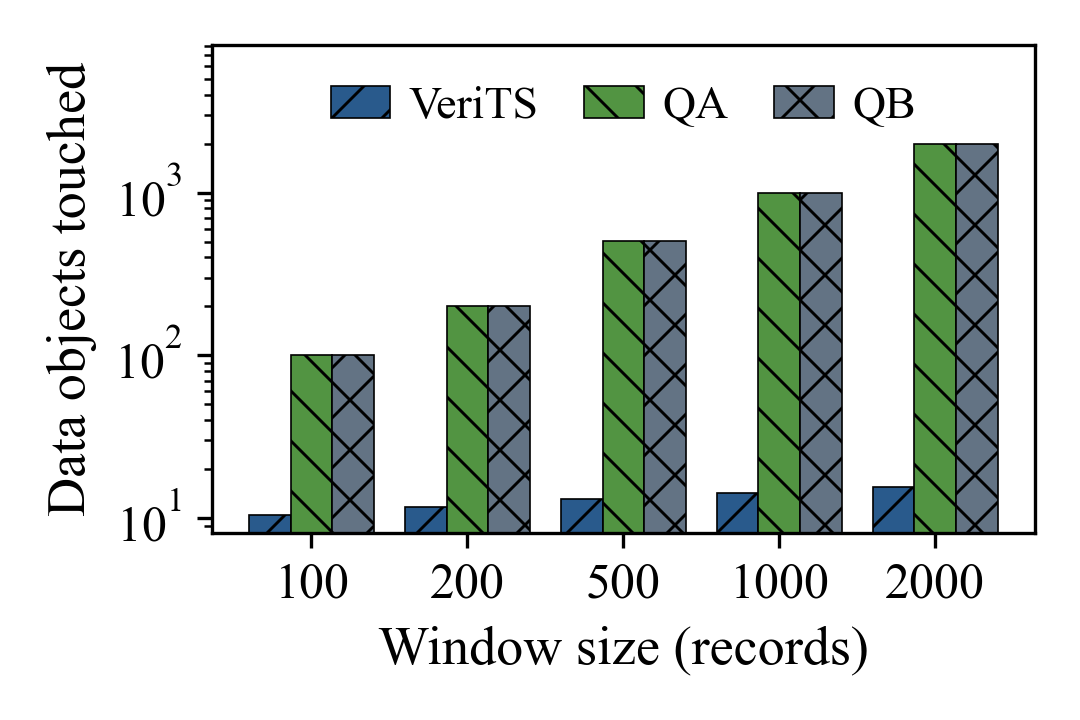}
  \caption{Objects touched}
\end{subfigure}\hfill
\begin{subfigure}{0.49\linewidth}
  \includegraphics[width=\linewidth]{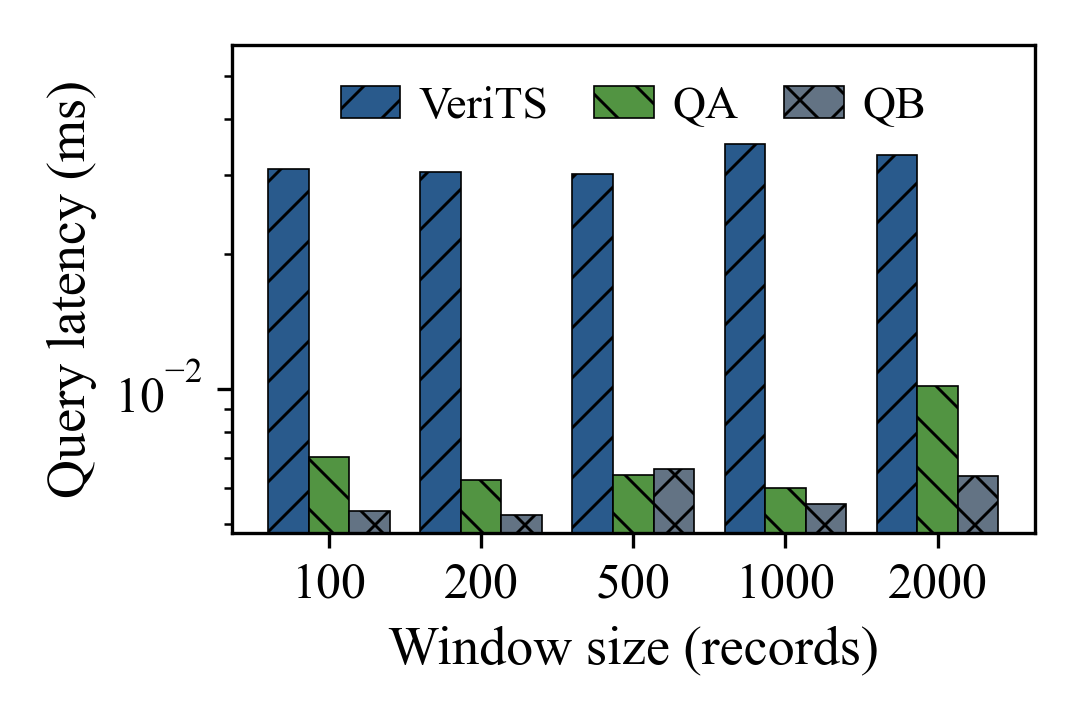}
  \caption{Query latency}
\end{subfigure}
\caption{Exact-path query cost on ETH-Gas.}
\vspace{-5mm}
\label{fig:exact-query}
\end{figure}

\subsection{Query Performance}\label{sec:exp-query}
Fig.~\ref{fig:exact-query}(a) reports the number of data objects each method must touch to answer a windowed aggregate (AIT nodes for \texttt{VeriTS}, records for a scan). \texttt{VeriTS} touches a near-constant handful of objects as the window grows from 100 to 2000 records, while QA and QB touch every record in the window. The gap reaches a factor of 130. This is because \texttt{VeriTS} folds the $O(\log n)$ node aggregates of a canonical cover rather than visiting the window's records, so a twentyfold longer window costs it five more nodes. Moreover, we make two observations. First, the wall-clock latency in Fig.~\ref{fig:exact-query}(b) does not reproduce this separation. QA and QB are several times faster than \texttt{VeriTS}, though all three stay well under a millisecond. This is because a scan is one vectorized pass while folding walks interpreted tree objects. The gap reflects the prototype rather than the structure, and the advantage survives in the transferred bytes, measured next. Second, the five streams give almost the same numbers. The VO at a 2000-record window varies by 10.9\% across all of them. This is because the exact path never inspects what a value means. Its cost is set by the canonical cover's shape, so a stream of near-constant dust and one of heavy-tailed transfers are indistinguishable to it.

Range queries need no separate experiment on the exact path. This is because the answer is the $k$ records themselves, so the cost is the unavoidable $O(k)$ result transfer of Section~\ref{sec:rangetag}. It appears as the Exact series of Figs.~\ref{fig:range-ev}--\ref{fig:range-be}(a), where it serves as the approximate path's baseline.

\begin{figure}[!t]
\centering
\begin{subfigure}{0.49\linewidth}
  \includegraphics[width=\linewidth]{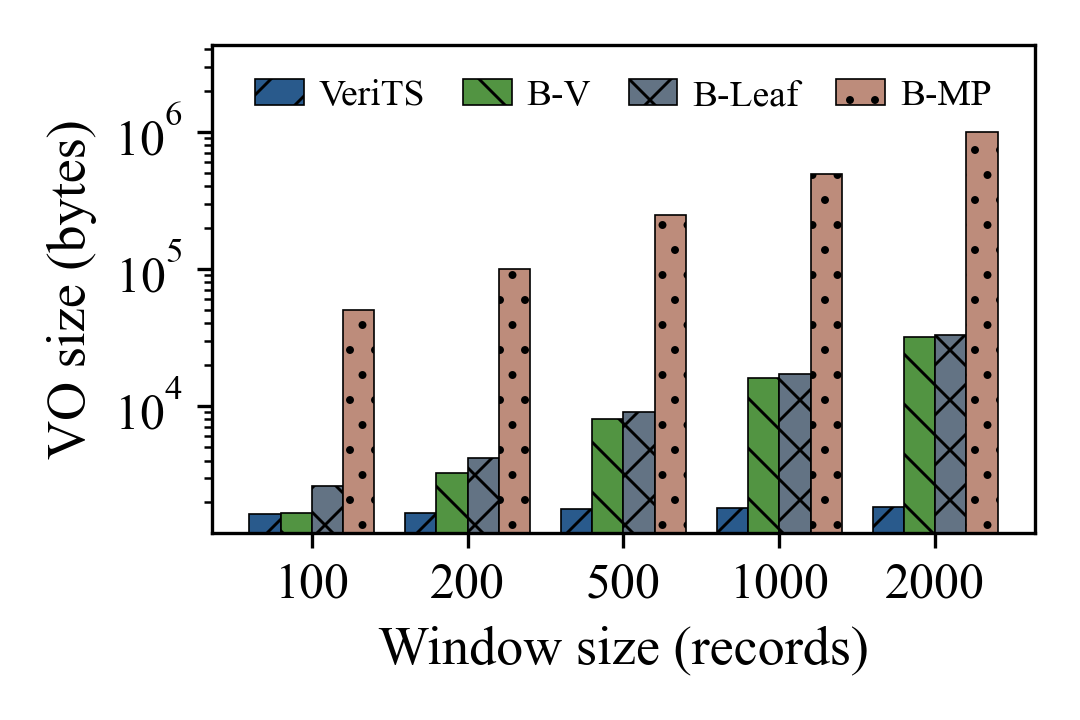}
  \caption{VO size}
\end{subfigure}\hfill
\begin{subfigure}{0.49\linewidth}
  \includegraphics[width=\linewidth]{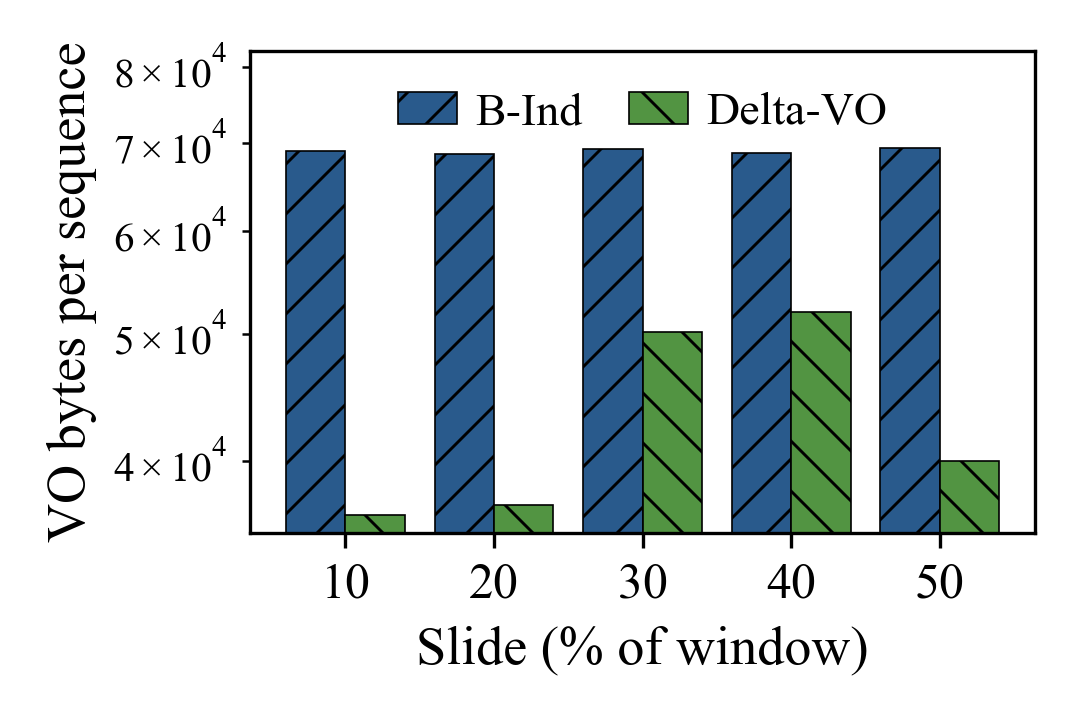}
  \caption{Sliding sequences}
\end{subfigure}
\caption{Exact-path verification cost on ETH-Gas.}
\vspace{-3mm}
\label{fig:exact-verify}
\end{figure}

\subsection{Verification Performance}\label{sec:exp-verify}

\begin{figure}[!t]
\centering
\begin{subfigure}{0.49\linewidth}
  \includegraphics[width=\linewidth]{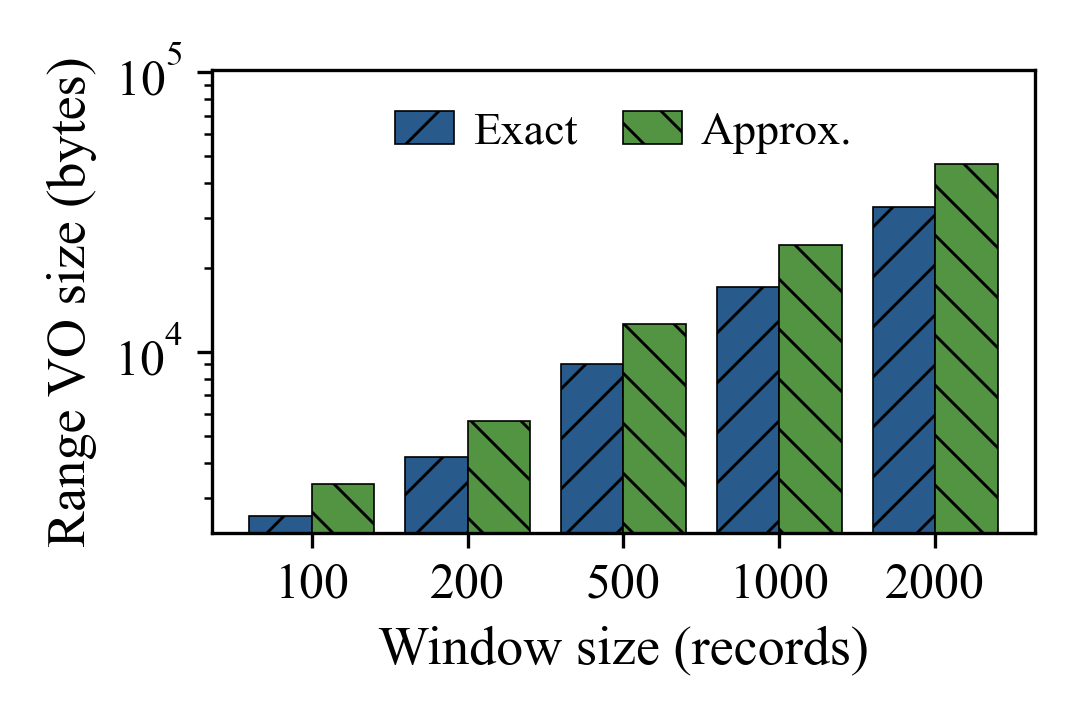}
  \caption{TRON-Transfer, event}
\end{subfigure}\hfill
\begin{subfigure}{0.49\linewidth}
  \includegraphics[width=\linewidth]{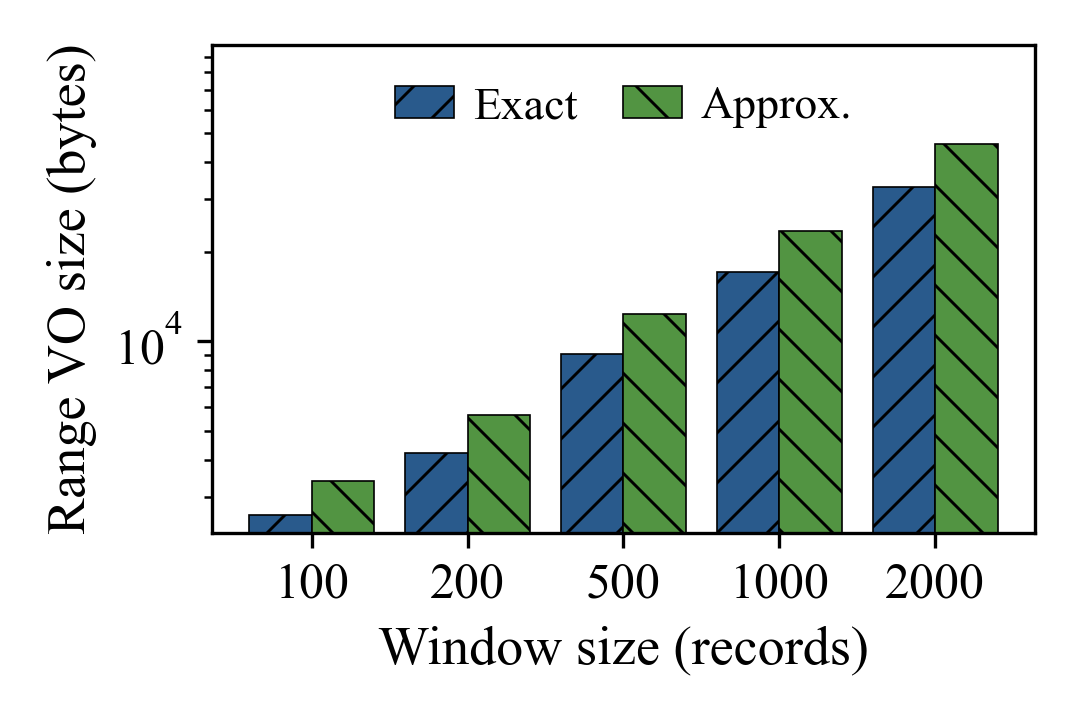}
  \caption{ETH-Transfer, event}
\end{subfigure}
\caption{Range-retrieval VO size (events).}
\vspace{-5mm}
\label{fig:range-ev}
\end{figure}

\subsubsection{Verification cost}\label{sec:exp-vcost}
Fig.~\ref{fig:exact-verify}(a) reports VO size for \texttt{VeriTS} and the three baselines. Over a 2000-record window \texttt{VeriTS} ships a KB-scale VO. B-V and B-Leaf ship roughly $18\times$ more, and B-MP $544\times$ more. This is because (i) B-V re-downloads every record so the client can recompute; (ii) B-Leaf ships every authenticated leaf so the client can rebuild the aggregate; (iii) B-MP pays a $\log n$-deep membership path per record, so its proof grows as $k\log n$; and (iv) \texttt{VeriTS} pays $\log n$ once for the whole window, folding the subtree aggregates that Eq.~\eqref{eq:node} already binds into the root. We make one observation. The advantage in client time is larger still, a factor of 1031 over B-Leaf at the longest window, with \texttt{VeriTS} still sub-millisecond. This is because B-Leaf must rehash every leaf in the window and fold the aggregate itself, which is linear in $k$ with a hash per record, whereas \texttt{VeriTS} rehashes only the $O(\log n)$ nodes of the cover.

\subsubsection{Sliding-window sequences}
Fig.~\ref{fig:exact-verify}(b) reports delta VOs over 40-window sliding sequences. On ETH-Gas the delta VO is $1.90\times$ smaller than B-Ind at a 10\% slide. This is because the client caches the $(\eta_u,\alpha_u)$ pairs it has already authenticated and the query layer omits them from later windows. We make one observation. The saving does not increase as the slide shrinks, although a 10\% slide overlaps far more records than a 50\% one. This is because a canonical cover is determined by a window's alignment rather than by its content, so two windows sharing most of their records can still decompose into largely different nodes. The reuse a delta VO can exploit is bounded by that re-alignment, not by record overlap.

\subsubsection{Tamper detection}\label{sec:exp-tamper}
We instantiate the adversary of Theorem~\ref{theo:verif} as nine concrete forgeries and replay each over twenty randomly placed windows of every stream. Four forgeries target the exact path. They are an inflated subtree aggregate, a forged node hash, a dropped covering subtree with the answer reduced to match, and a modified boundary record. Five more target what is new on the approximate path. They are a shifted value model $\theta_j$, a halved declared $\varepsilon^{v}_j$, an omitted segment, a narrowed answer interval, and a halved anchored budget cap. All 900 forgeries were rejected. This is because every field the client reads is bound into the anchored root by Eq.~\eqref{eq:node}, and the two quantities that live outside any single node, the folded aggregate and the assembled interval, are recomputed by the client rather than taken from the query layer. We make one observation. The three attacks that present an internally consistent VO are each caught by a different check. They are the covering-set test, the root check, and the client's own re-derivation. This is why $\varepsilon$-soundness needs the budget cap anchored rather than reported alongside the answer.

\subsubsection{Scalability and construction cost}\label{sec:exp-scale}
We sweep stream length by truncating TRON-Dust to prefixes of 10\,k, 20\,k, 50\,k, 100\,k and 168{,}499 records, holding the window at 1{,}000 records. The VO grows by 17\% and client time by 26\% for a 16.8-fold longer stream. This is because only the depth of the structure changes. A fixed window's canonical cover has the same shape however deep the tree is, so each doubling of $n$ adds one sibling to the multiproof and one hash to the client's fold. Miner-side certification, meanwhile, costs microseconds per record. That is a factor of 18 to 462 below the encoder that produced the segments, a cost no miner pays. This is because certification replays the two closed forms $\hat v_j$ and $\hat t_j$ at each position rather than the search that fitted them, which is the asymmetry Section~\ref{sec:segauth} relies on to keep the encoder out of the trust path.

\begin{figure}[!t]
\centering
\begin{subfigure}{0.49\linewidth}
  \includegraphics[width=\linewidth]{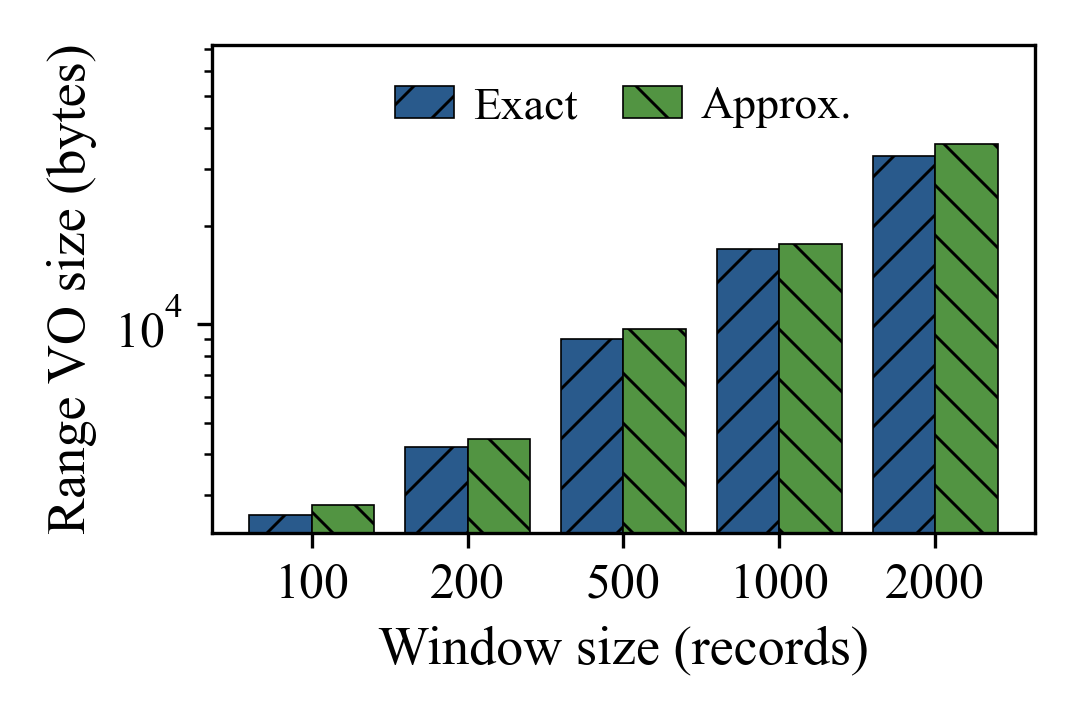}
  \caption{ETH-Gas, rate}
\end{subfigure}\hfill
\begin{subfigure}{0.49\linewidth}
  \includegraphics[width=\linewidth]{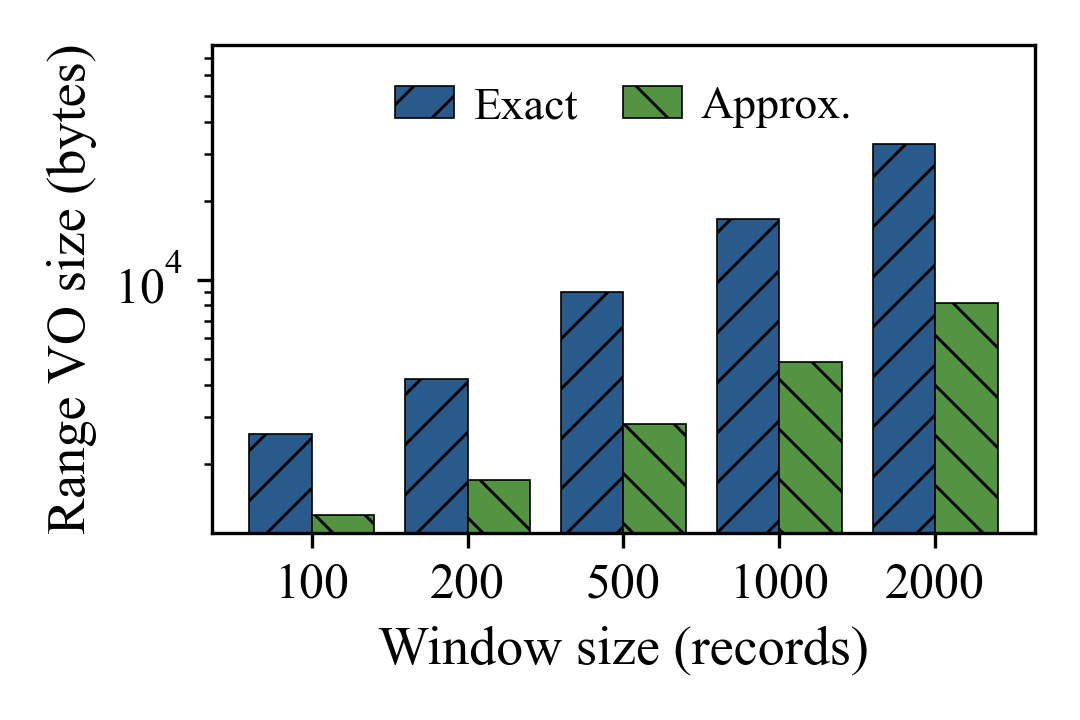}
  \caption{TRON-Balance, state}
\end{subfigure}
\caption{Range-retrieval VO size (rate and state).}
\vspace{-5mm}
\label{fig:range-rs}
\end{figure}

\subsection{Approximate Path Performance}\label{sec:exp-fast}

\subsubsection{Range retrieval}
Figs.~\ref{fig:range-ev}--\ref{fig:range-be}(a) report range-query VO size for the exact and approximate paths on all five streams, ordered by compressibility. The approximate path ships more than the exact path on the two event streams, roughly the same on the rate stream, and $4.05\times$ smaller on the state stream. This is because a range query obliges the exact path to ship the $k$ records it returns while the approximate path ships the $k/(n/s)$ segments covering them, so the trade turns on whether a segment's model fields cost less than the records they stand for. Moreover, we make one observation. The outcome tracks the kind of quantity rather than the chain. This is because event amounts are near-independent draws with no local structure for a line to follow, so their segments stay two records long. A running total moves smoothly enough for one line to cover thirteen records. A constant is covered by one line indefinitely. The two TRON streams isolate the effect, since TRON-Balance is the running total of exactly the transfers in TRON-Transfer. The two streams hold identical data, and they receive opposite verdicts.

\begin{figure}[!t]
\centering
\begin{subfigure}{0.49\linewidth}
  \includegraphics[width=\linewidth]{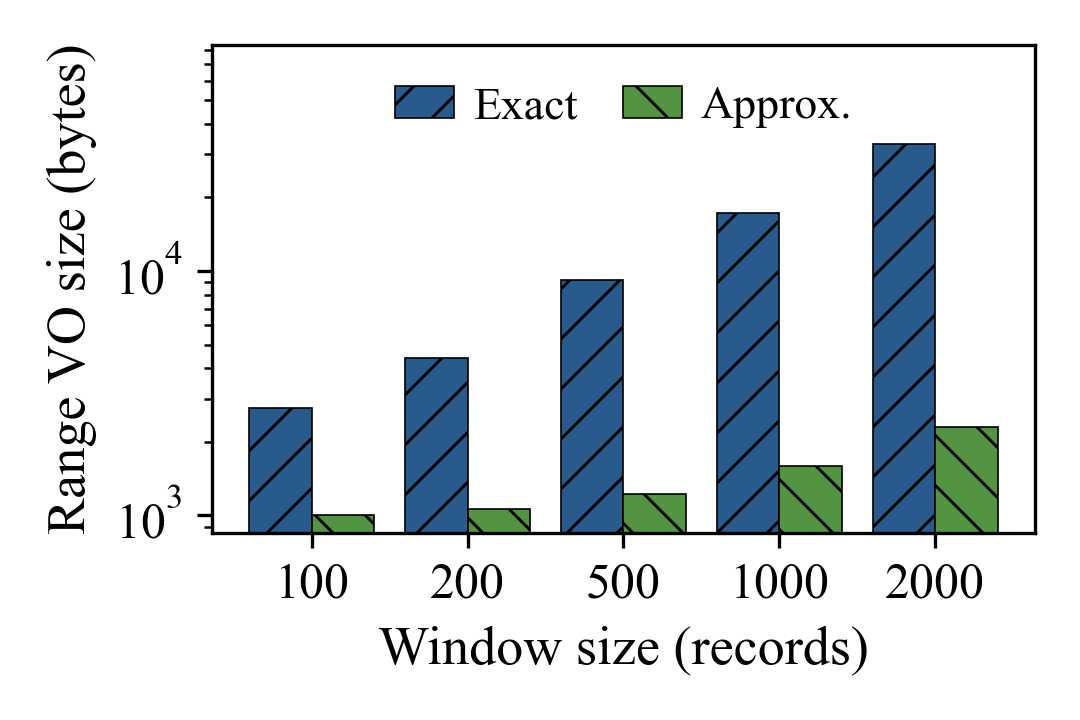}
  \caption{TRON-Dust, degenerate}
\end{subfigure}\hfill
\begin{subfigure}{0.49\linewidth}
  \includegraphics[width=\linewidth]{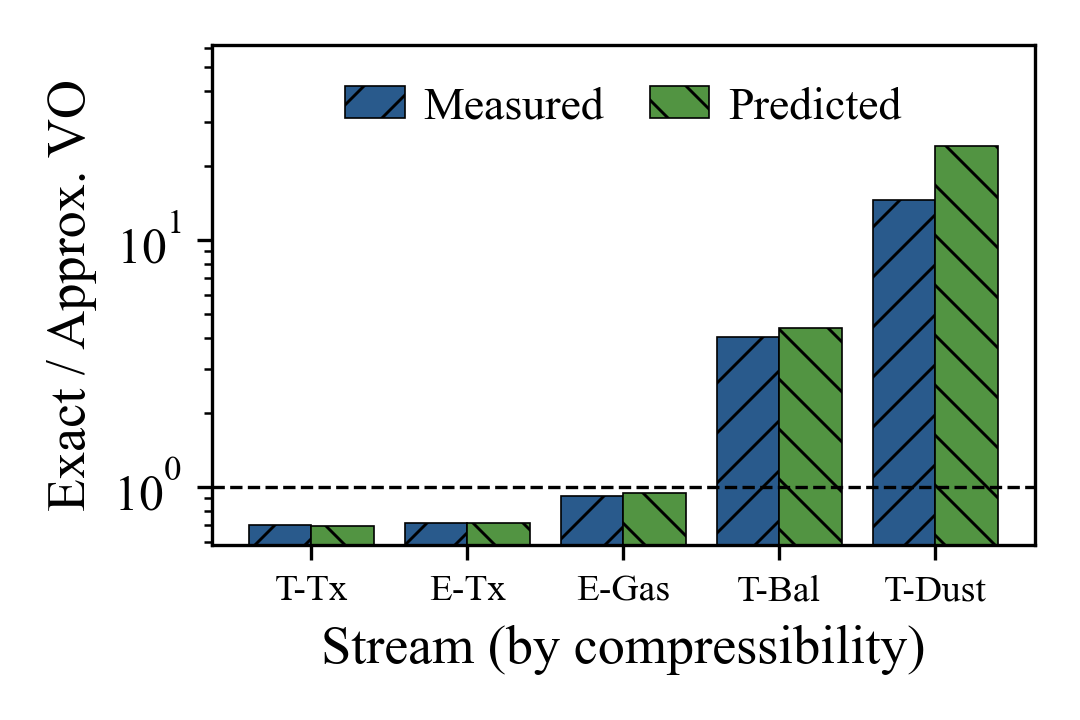}
  \caption{Break-even}
\end{subfigure}
\caption{Range-retrieval VO size and the break-even point.}
\vspace{-3mm}
\label{fig:range-be}
\end{figure}

\begin{figure}[!t]
\centering
\begin{subfigure}{0.49\linewidth}
  \includegraphics[width=\linewidth]{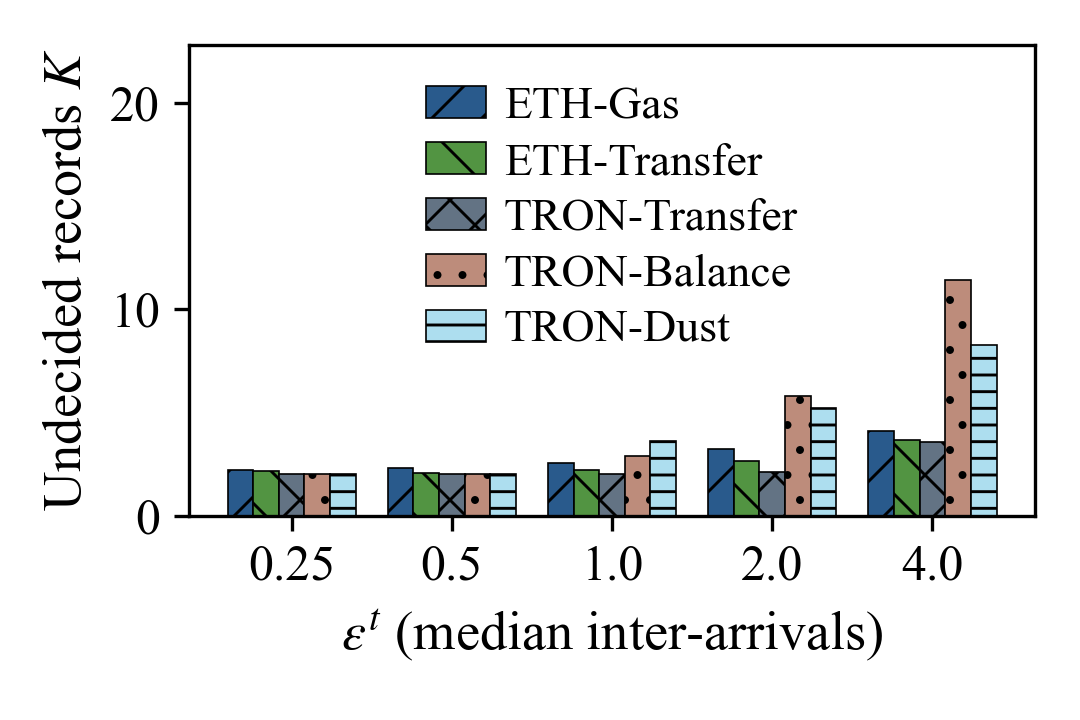}
  \caption{Boundary cost}
\end{subfigure}\hfill
\begin{subfigure}{0.49\linewidth}
  \includegraphics[width=\linewidth]{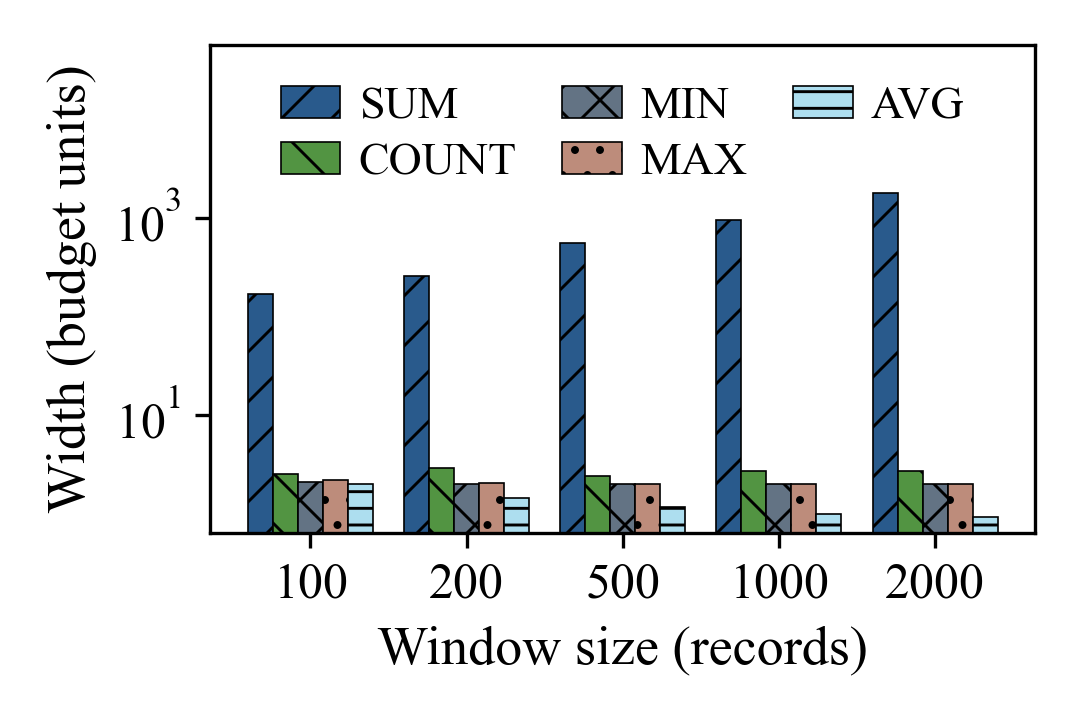}
  \caption{Five aggregates}
\end{subfigure}
\caption{The cost of an $\varepsilon$-bounded answer.}
\vspace{-5mm}
\label{fig:fastcost}
\end{figure}

\subsubsection{The break-even point}
Fig.~\ref{fig:range-be}(b) puts the measured saving next to what Eq.~\eqref{eq:payoff} predicts from byte counts alone, $n/s$ against $b_g/b_r=3$. The dashed line marks parity. The prediction is accurate on the four non-degenerate streams. It gives $0.70$ against a measured $0.70$, $0.71$ against $0.71$, $0.94$ against $0.92$, and $4.39$ against $4.05$. This is because the break-even is a statement about what a proof carries, and $n/s$ is a property of the data, so the crossing point is known before any query runs. A record costs 16 bytes on the wire and a segment's model fields cost 48. We make one observation. On TRON-Dust the prediction of $24.0$ overshoots the measured $14.5$. This is because the segment fields no longer dominate a proof once a stream compresses seventyfold. The Merkle siblings binding segments to the anchored root do not shrink with $n/s$, so the saving saturates.

\subsubsection{Aggregation on the approximate path}
Under the same segmentations, the approximate path moves aggregation VO size by well under a factor of two on every stream. This is because compressing a stream by $n/s$ removes only $\log_2(n/s)$ levels from an already-logarithmic proof, about three at tenfold compression, as Section~\ref{sec:payoff} predicts. Model segments therefore pay off on retrieval rather than aggregation.

\subsubsection{Boundary resolution}
Fig.~\ref{fig:fastcost}(a) reports the undecided count $K$ at a window edge as the arrival budget sweeps from a quarter to four times a stream's median inter-arrival time. $K$ rises with the budget, staying in single digits over most of the range and reaching about eleven only at the loosest budget. This is because $k_j=\lceil 2\varepsilon^{t}_j/\mu_j\rceil$ counts the records the arrival model cannot separate from the edge, a count proportional to the budget in inter-arrivals. Across every stream, budget and window in this experiment, 1{,}000 queries in all, the verified interval contained the exact answer.

\subsubsection{The five aggregates}\label{sec:exp-aggs}
Fig.~\ref{fig:fastcost}(b) reports the verified interval width of all five aggregates on ETH-Gas, in units of the value budget $\varepsilon^{v}$. As the window grows from 100 to 2{,}000 records the SUM interval widens from 170 to 1{,}830 budget units, while MIN and MAX stay within 2.0 to 2.2 and AVG narrows from 2.0 to 0.93. COUNT, whose interval is the undecided run itself, stays between 2.5 and 3.0 records. This is because Lemma~\ref{lem:eps} charges SUM one $\varepsilon^{v}$ for every record the window covers, whereas an extremum pays the budget once however many records there are, and AVG divides a SUM interval that grows with the window by a count that grows with it too. The excess of MIN and MAX over 2.0 appears only at the smallest window, where an undecided candidate holds the extremum. Moreover, we make one observation. At 100 records SUM is charged 1.7 budget units per covered record, close to AVG's whole interval of 2.0, and at 2{,}000 the SUM interval is 1{,}961 times the AVG interval. This is because the boundary term $\Delta_t$ is paid twice regardless of window length while the value term $\Delta_v$ is paid once per record, so short windows are dominated by the edges for every aggregate and only SUM inherits the growth of $\Delta_v$. Across five streams and five window lengths, 500 windows and 2{,}500 interval checks in all, every one of the five intervals contained the exact answer.

\begin{figure}[!t]
\centering
\begin{subfigure}{0.49\linewidth}
  \includegraphics[width=\linewidth]{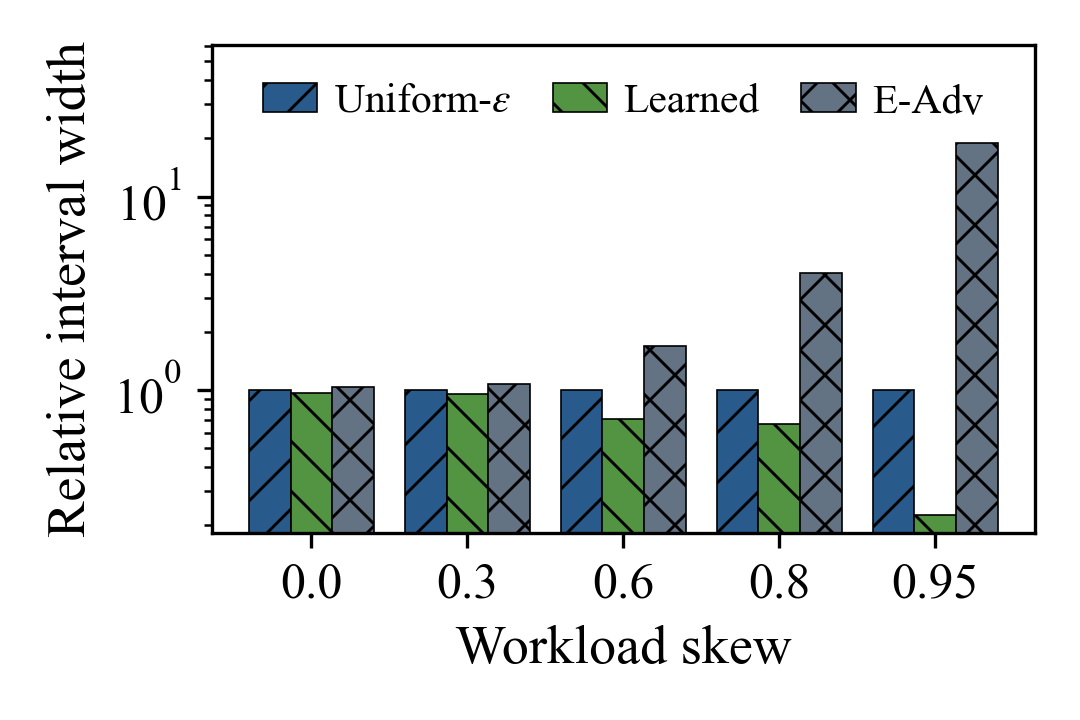}
  \caption{ETH-Gas, rate}
\end{subfigure}\hfill
\begin{subfigure}{0.49\linewidth}
  \includegraphics[width=\linewidth]{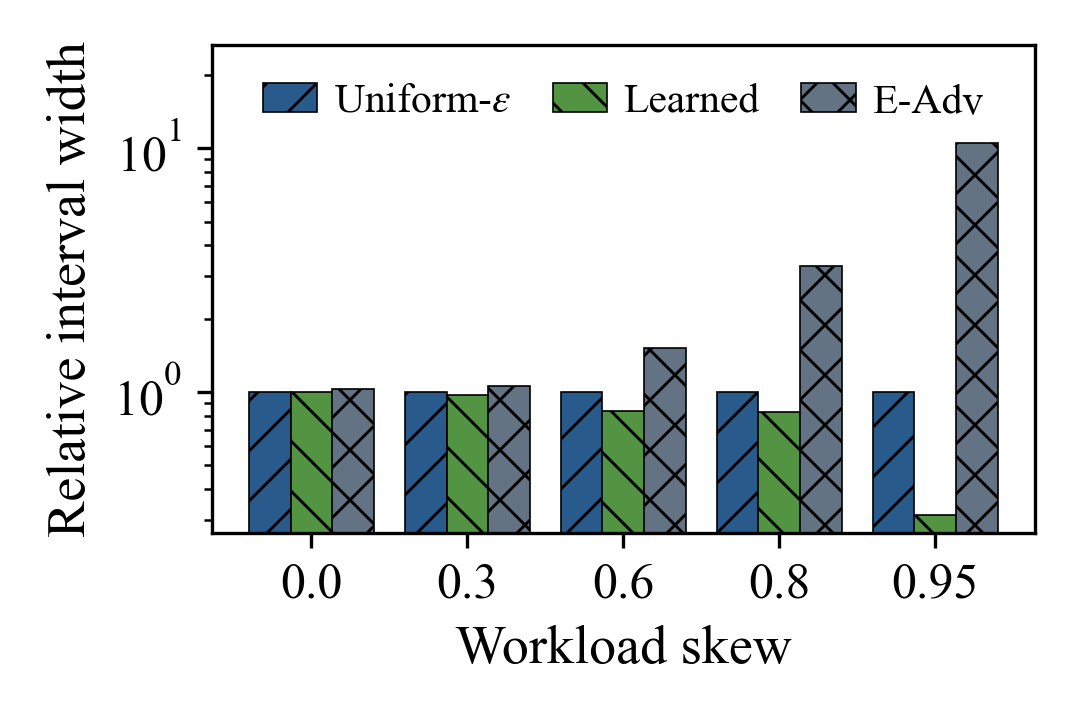}
  \caption{TRON-Balance, state}
\end{subfigure}
\caption{Interval width against skew (rate and state).}
\vspace{-3mm}
\label{fig:encoder-rs}
\end{figure}

\begin{figure}[!t]
\centering
\begin{subfigure}{0.49\linewidth}
  \includegraphics[width=\linewidth]{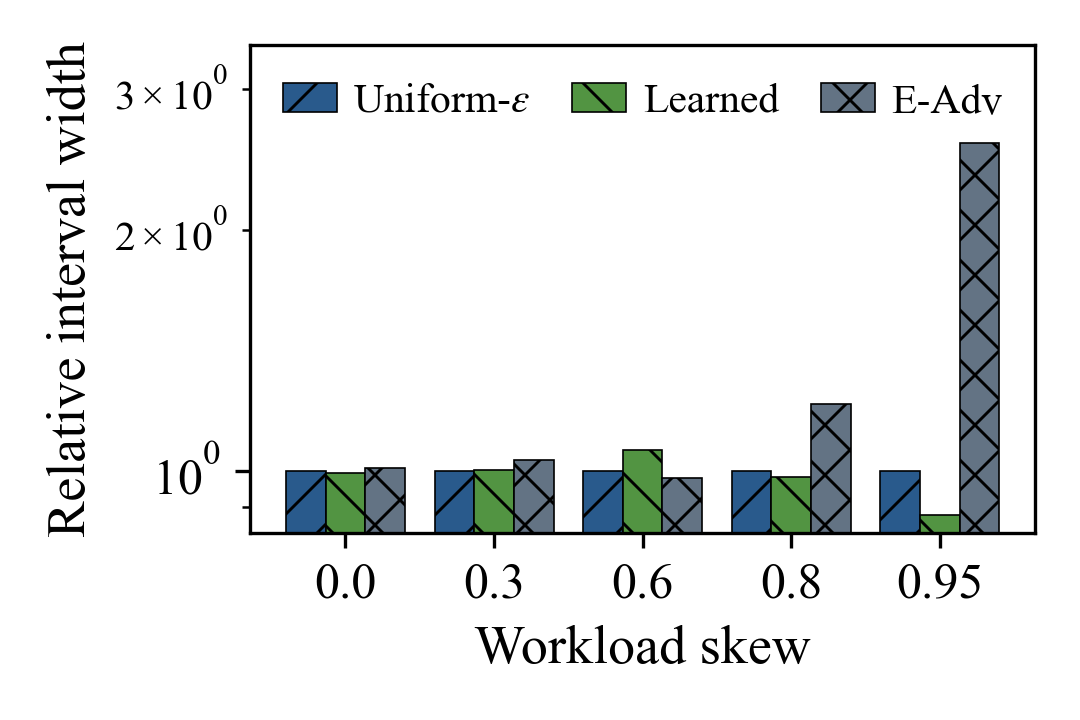}
  \caption{TRON-Transfer, event}
\end{subfigure}\hfill
\begin{subfigure}{0.49\linewidth}
  \includegraphics[width=\linewidth]{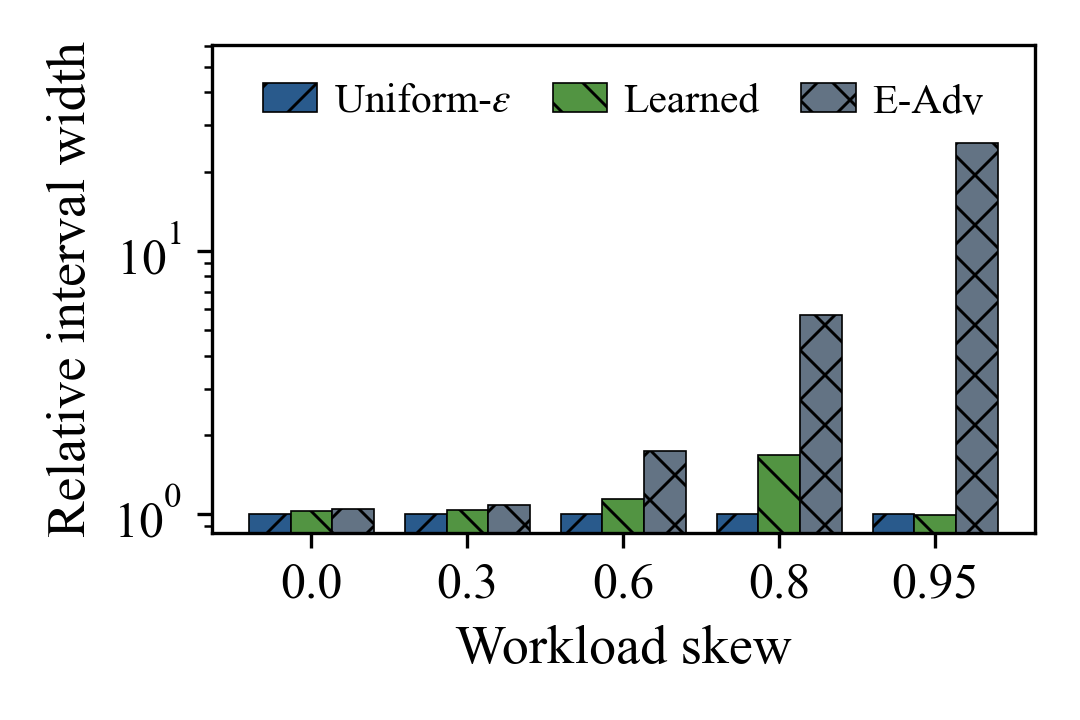}
  \caption{TRON-Dust, degenerate}
\end{subfigure}
\caption{Interval width against skew (events and degenerate).}
\vspace{-5mm}
\label{fig:encoder-ed}
\end{figure}

\subsubsection{Encoder study}\label{sec:exp-encoder}
Figs.~\ref{fig:encoder-rs} and~\ref{fig:encoder-ed} report the verified interval width under each policy relative to Uniform-$\varepsilon$ as the workload skew grows from 0 to 0.95. The three policies differ only in how they split $\varepsilon^{v}$ and $\varepsilon^{t}$ between interior and edge segments, the allocation question of Section~\ref{sec:payoff}. E-Adv widens the interval by up to $25.6\times$ on TRON-Dust, $18.9\times$ on ETH-Gas and $10.5\times$ on TRON-Balance, while Learned narrows it to $0.23$ and $0.31$ of the uniform width on the latter two. Under every encoder and skew, on these four streams, the returned interval contained the exact answer, $108{,}000$ checks over the five aggregates with no exception. This is because the client's checks read only authenticated fields whose bounds miners validated against the chain, so an encoder chooses how much slack to declare but never whether the declaration holds. This confirms Proposition~\ref{prop:encoder} empirically. Moreover, we make two observations. First, Learned's benefit tracks the value budget's share of the interval, not the stream's compressibility. Halving $\varepsilon^{v}$ alone and re-encoding takes the width to $0.42$ on ETH-Gas and $0.62$ on TRON-Balance, where the value term is $79.8\%$ and $76.0\%$ of the interval, but to $0.99$ on TRON-Transfer and $1.00$ on TRON-Dust, where it is $6.6\%$ and $0.0\%$. On ETH-Gas the drop beats the linear $0.60$ its share predicts. This is because re-encoding at the halved budget cuts the declared residuals $3.5\times$. On TRON-Dust only the arrival budget moves the width, and halving $\varepsilon^{t}$ takes it to $0.58$. This is because the water-filling law is derived under a budget set by the value bound, so where a single model fits the values exactly the policy reallocates a quantity the answer does not depend on. Compressibility predicts none of this. TRON-Dust is the most compressible stream yet the one on which Learned never beats \mbox{Uniform-$\varepsilon$} (Fig.~\ref{fig:encoder-ed}(b)). Second, E-Adv's damage grows with skew and Learned's benefit does not appear until the skew is substantial. At moderate skew either policy can land on the wrong side of parity. This is because a concentrated workload creates both the opportunity and the exposure. The more queries pile onto one region, the more an allocator gains by tightening there, and the more it loses by tightening everywhere else.

% ======================================================================
\section{Conclusion}\label{sec:conclusion}
In this paper, we propose the \texttt{VeriTS} framework that enables efficient and verifiable time-series query services for blockchain systems. \texttt{VeriTS} maintains an off-chain query layer that stores each stream's records and model segments in authenticated aggregate interval trees. \texttt{VeriTS} verifies completeness through a minimum covering set and soundness through aggregate folding, and extends both guarantees to an approximate path over model segments under a certified error bound. Experiments offer evidence that \texttt{VeriTS} supports efficient verifiable query services while keeping every accepted answer within its certified interval, even under an adversarial encoder. In future research, it is of interest to support more complex verifiable time-series queries, such as pattern matching, similarity search, top-$k$ ranking, and forecasting.

\bibliographystyle{IEEEtran}
\bibliography{vtq}

@article{treleaven2017blockchain,
  title={Blockchain technology in finance},
  author={Treleaven, Philip and Brown, Richard Gendal and Yang, Danny},
  journal={Computer},
  volume={50},
  number={9},
  pages={14--17},
  year={2017}
}

@article{attaran2022blockchain,
  title={Blockchain technology in healthcare: Challenges and opportunities},
  author={Attaran, Mohsen},
  journal={IJHM},
  volume={15},
  number={1},
  pages={70--83},
  year={2022}
}

@inproceedings{zhu2019sebdb,
  title={{SEBDB}: Semantics empowered blockchain database},
  author={Zhu, Yanchao and Zhang, Zhao and Jin, Cheqing and Zhou, Aoying and Yan, Ying},
  booktitle={ICDE},
  pages={1820--1831},
  year={2019},
}

@article{wu2021vql,
  title={{VQL}: Efficient and verifiable cloud query services for blockchain systems},
  author={Wu, Haotian and Peng, Zhe and Guo, Songtao and Yang, Yuanyuan and Xiao, Bin},
  journal={TPDS},
  volume={33},
  number={6},
  pages={1393--1406},
  year={2022},
}

@inproceedings{xu2019vchain,
  title={{vChain}: Enabling verifiable {B}oolean range queries over blockchain databases},
  author={Xu, Cheng and Zhang, Ce and Xu, Jianliang},
  booktitle={SIGMOD},
  pages={141--158},
  year={2019}
}

@inproceedings{wang2022vchain+,
  title={{vChain+}: Optimizing verifiable blockchain {B}oolean range queries},
  author={Wang, Haixin and Xu, Cheng and Zhang, Ce and Xu, Jianliang and Peng, Zhe and Pei, Jian},
  booktitle={ICDE},
  pages={1927--1940},
  year={2022}
}

@inproceedings{li2017etherql,
  title={{EtherQL}: A query layer for blockchain system},
  author={Li, Yang and Zheng, Kai and Yan, Ying and Liu, Qi and Zhou, Xiaofang},
  booktitle={DASFAA},
  pages={556--567},
  year={2017},
}

@inproceedings{peng2020falcondb,
  title={{FalconDB}: Blockchain-based collaborative database},
  author={Peng, Yanqing and Du, Min and Li, Feifei and Cheng, Raymond and Song, Dawn},
  booktitle={SIGMOD},
  pages={637--652},
  year={2020}
}

@inproceedings{zhang2019gem,
  title={{{GEM$^2$-tree}}: A gas-efficient structure for authenticated range queries in blockchain},
  author={Zhang, Ce and Xu, Cheng and Xu, Jianliang and Tang, Yuzhe and Choi, Byron},
  booktitle={ICDE},
  pages={842--853},
  year={2019},
}

@inproceedings{zhang2021authenticated,
  title={Authenticated keyword search in scalable hybrid-storage blockchains},
  author={Zhang, Ce and Xu, Cheng and Wang, Haixin and Xu, Jianliang and Choi, Byron},
  booktitle={ICDE},
  pages={996--1007},
  year={2021},
}

@article{ruan2021lineagechain,
  title={{LineageChain}: A fine-grained, secure and efficient data provenance system for blockchains},
  author={Ruan, Pingcheng and Dinh, Tien Tuan Anh and Lin, Qian and Zhang, Meihui and Chen, Gang and Ooi, Beng Chin},
  journal={VLDBJ},
  volume={30},
  number={1},
  pages={3--24},
  year={2021},
}

@article{singh2023efficient,
  title={Efficient and lightweight indexing approach for multi-dimensional historical data in blockchain},
  author={Singh, Bikash Chandra and Ye, Qingqing and Hu, Haibo and Xiao, Bin},
  journal={FGCS},
  volume={139},
  pages={210--223},
  year={2023}
}

@inproceedings{zhang2015integridb,
  title={{IntegriDB}: Verifiable {SQL} for outsourced databases},
  author={Zhang, Yupeng and Katz, Jonathan and Papamanthou, Charalampos},
  booktitle={CCS},
  pages={1480--1491},
  year={2015}
}

@article{zhou2023veridkg,
  title={{VeriDKG}: A verifiable {SPARQL} query engine for decentralized knowledge graphs},
  author={Zhou, Enyuan and Guo, Song and Hong, Zicong and Jensen, Christian S. and Xiao, Yang and Zhang, Dalin and Liang, Jinwen and Pei, Qingqi},
  journal={PVLDB},
  volume={17},
  number={4},
  pages={912--925},
  year={2023},
}

@article{liu2026clustered,
  title={Clustered federated learning for energy-harvesting smart meters in {P2P} energy trading},
  author={Liu, Ziming and Huang, Bonan and Zhang, Cheng and Yao, Zhongming and Li, Tianyi and Sun, Qiuye and Li, Yushuai},
  journal={TGCN},
  volume={10},
  pages={585--596},
  year={2025},
}

@inproceedings{li2020survey,
  title={A survey on blockchain anomaly detection using data mining techniques},
  author={Li, Ji and Gu, Chunxiang and Wei, Fushan and Chen, Xi},
  booktitle={BlockSys},
  pages={491--504},
  year={2019}
}

@inproceedings{li2022evolutionary,
  title={Evolutionary clustering of moving objects},
  author={Li, Tianyi and Chen, Lu and Jensen, Christian S. and Pedersen, Torben Bach and Gao, Yunjun and Hu, Jilin},
  booktitle={ICDE},
  pages={2399--2411},
  year={2022}
}

@article{li2020compression,
  title={Compression of uncertain trajectories in road networks},
  author={Li, Tianyi and Huang, Ruikai and Chen, Lu and Jensen, Christian S. and Pedersen, Torben Bach},
  journal={PVLDB},
  volume={13},
  number={7},
  pages={1050--1063},
  year={2020},
}

@article{li2021trace,
  title={{TRACE}: Real-time compression of streaming trajectories in road networks},
  author={Li, Tianyi and Chen, Lu and Jensen, Christian S. and Pedersen, Torben Bach},
  journal={PVLDB},
  volume={14},
  number={7},
  pages={1175--1187},
  year={2021},
}

@article{li2010authenticated,
  title={Authenticated index structures for aggregation queries},
  author={Li, Feifei and Hadjieleftheriou, Marios and Kollios, George and Reyzin, Leonid},
  journal={TISSEC},
  volume={13},
  number={4},
  pages={1--35},
  year={2010}
}

@inproceedings{wang2023iotdb,
  author  = {Wang, Chen and Qiao, Jialin and Huang, Xiangdong and Song, Shaoxu and Hou, Haonan and Jiang, Tian and Rui, Lei and Wang, Jianmin and Sun, Jiaguang},
  title   = {{Apache} {IoTDB}: A time series database for {IoT} applications},
  booktitle = {SIGMOD},
  volume  = {1},
  number  = {2},
  pages   = {195:1--195:27},
  year    = {2023}
}

@article{eichinger2015time,
  author = {Eichinger, Frank and Efros, Pavel and Karnouskos, Stamatis and B{\"o}hm, Klemens},
  title = {A time-series compression technique and its application to the smart grid},
  journal = {VLDBJ},
  volume = {24},
  number = {2},
  pages = {193--218},
  year = {2015},
}

@inproceedings{yao2025vgq,
  title = {{VGQ}: Enabling verifiable graph queries on blockchain systems},
  author = {Yao, Zhongming and Li, Tianyi and Xin, Junchang and Li, Yushuai and Wang, Chenxu and Wang, Zhiqiong and Srivastava, Divesh and Jensen, Christian S.},
  booktitle = {ICDE},
  pages = {3602--3614},
  year = {2025},
}

@article{yao2026vtrq,
  title = {{VTRQ}: Enabling verifiable trajectory range queries in hybrid-storage blockchains},
  author = {Yao, Zhongming and Xin, Junchang and Song, Yumeng and Mao, Yusen and Torp, Kristian and Ding, Yuemin and Srivastava, Divesh and Li, Yushuai and Jensen, Christian S. and Li, Tianyi},
  journal = {arXiv preprint arXiv:2608.21314},
  year = {2026},
}

@article{xie2014optimalplr,
  author = {Xie, Qing and Pang, Chaoyi and Zhou, Xiaofang and Zhang, Xiangliang and Deng, Ke},
  title = {Maximum error-bounded piecewise linear representation for online stream approximation},
  journal = {VLDBJ},
  volume = {23},
  number = {6},
  pages = {915--937},
  year = {2014},
}

@inproceedings{yao2024camel,
  author = {Yao, Yuanyuan and Chen, Lu and Fang, Ziquan and Gao, Yunjun and Jensen, Christian S and Li, Tianyi},
  title = {{Camel}: Efficient compression of floating-point time series},
  booktitle = {SIGMOD},
  volume = {2},
  number = {6},
  pages = {1--26},
  year = {2024},
}

@article{hu2024estimator,
  author={Hu, Danlei and Fang, Ziquan and Fang, Hanxi and Li, Tianyi and Shen, Chunhui and Chen, Lu and Gao, Yunjun},
  title={{Estimator}: An effective and scalable framework for transportation mode classification over trajectories},
  journal={TITS},
  volume={25},
  number={11},
  pages={15562--15573},
  year={2024},
}

@inproceedings{yao2024tsec,
  title={{Tsec}: An efficient and effective framework for time series classification},
  author={Yao, Yuanyuan and Jie, Hailiang and Chen, Lu and Li, Tianyi and Gao, Yunjun and Wen, Shiting},
  booktitle={ICDE},
  pages={1394--1406},
  year={2024},
}

@article{yao2023learned,
  title={Learned-index-based semantic keyword query on blockchain},
  author={Yao, Zhongming and Xin, Junchang and Hao, Kun and Wang, Zhiqiong and Zhu, Wancheng},
  journal={Mathematics},
  volume={11},
  number={9},
  pages={2055},
  year={2023},
}

@inproceedings{yao2023efficient,
  title={Efficient blockchain data trusty provenance based on the {W3C} {PROV} model},
  author={Yao, Zhongming and Wang, Zhiqiong and Wen, Liang and Hao, Kun and Xu, Junming},
  booktitle={ADMA},
  pages={61--76},
  year={2023},
}

@inproceedings{wen2023fine,
  title={A Fine-Grained Verification Method for Blockchain Data Based on {M}erkle Path Sharding},
  author={Wen, Liang and Wang, Zhiqiong and Cui, Tingyu and Shi, Caiyun and Li, Baoting and Yao, Zhongming},
  booktitle={ADMA},
  pages={620--632},
  year={2023},
}

@inproceedings{xu2026tcrl,
  title={{TCRL}: Temporal-coupled adversarial training for robust constrained reinforcement learning in worst-case scenarios},
  author={Xu, Wentao and Yao, Zhongming and Li, Weihao and Song, Zhenghang and Song, Yumeng and Li, Tianyi and Li, Yushuai},
  booktitle={AAMAS},
  pages={3489--3491},
  year={2026},
}

@article{hao2022efficient,
  title={On efficient top-k transaction path query processing in blockchain database},
  author={Hao, Kun and Xin, Junchang and Wang, Zhiqiong and Yao, Zhongming and Wang, Guoren},
  journal={DKE},
  volume={141},
  pages={102079},
  year={2022},
}

@article{hao2023efficient,
  title={Efficient and secure data sharing scheme on interoperable blockchain database},
  author={Hao, Kun and Xin, Junchang and Wang, Zhiqiong and Yao, Zhongming and Wang, Guoren},
  journal={TBD},
  volume={9},
  number={4},
  pages={1171--1185},
  year={2023},
}

@article{hu2023spatio,
  author = {Hu, Danlei and Chen, Lu and Fang, Hanxi and Fang, Ziquan and Li, Tianyi and Gao, Yunjun},
  title = {Spatio-temporal trajectory similarity measures: A comprehensive survey and quantitative study},
  journal = {TKDE},
  volume = {36},
  number = {5},
  pages = {2191--2212},
  year = {2023},
}
\end{document}